\documentclass[submission,copyright,creativecommons]{eptcs}
\providecommand{\event}{TERMGRAPH 2020} %
\pdfoutput=1
\usepackage{underscore}           %

\usepackage{graphicx}
\usepackage[table]{xcolor}
\usepackage{array}
\usepackage{amsmath, amssymb,amsthm}
\usepackage{mathtools}

\theoremstyle{plain}

\newtheorem{lemma}{Lemma}
\newtheorem{theorem}{Theorem}

\theoremstyle{definition}
\newtheorem{definition}{Definition}
\newtheorem{example}{Example}

\newtheorem{assumption}{Assumption}{\bfseries}{\itshape}

\makeatletter
\def\bstr{b}
\def\bfstr{bf}
\def\cstr{c}
\def\fstr{f}
\def\strLst{A,B,C,D,d,E,F,G,H,I,J,K,L,M,N,O,P,Q,R,S,T,U,V,W,X,Y,Z}

\newcommand{\MkB}[1]{\expandafter\def\csname\bstr#1\endcsname{\mathbb{#1}}}
  \@for\i:=\strLst\do{%
    \expandafter\MkB \i     }

\newcommand{\MkBF}[1]{\expandafter\def\csname\bfstr#1\endcsname{\mathbf{#1}}}
  \@for\i:=\strLst\do{%
    \expandafter\MkBF \i     }

\newcommand{\MkCal}[1]{\expandafter\def\csname\cstr#1\endcsname{\mathcal{#1}}}
  \@for\i:=\strLst\do{%
    \expandafter\MkCal \i     }

\newcommand{\MkFrak}[1]{\expandafter\def\csname\fstr#1\endcsname{\mathfrak{#1}}}
  \@for\i:=\strLst\do{%
    \expandafter\MkFrak \i     }
    
\makeatother

\newcommand{\Lin}[1]{\mathop{\mathsf{Lin}}(#1)}

\newcommand{\LinEq}[1]{\overline{\mathsf{Lin}}(#1)_{\sim}}
\newcommand{\LinAc}[1]{\overline{\mathsf{Lin}}(#1)}

\newcommand{\bra}[1]{\left\langle #1\right\vert}
\newcommand{\ket}[1]{\left\vert #1\right\rangle}
\newcommand{\braket}[2]{\left\langle \left. #1 \right\vert #2\right\rangle}

\newcommand{\cond}[1]{\mathsf{cond}(#1)}

\newcommand{\ac}[1]{\mathsf{#1}} %

\newcommand{\Shift}{\mathsf{Shift}}

\newcommand{\Trans}{\mathsf{Trans}}

\newcommand{\MatchGT}[3]{\mathsf{M}^{{\text{\tiny $#1$}}}_{#2}(#3)}

\newcommand{\obj}[1]{\mathsf{obj}(#1)}

\newcommand{\mIO}{\mathop{\varnothing}}

\newcommand{\canRep}[2]{\overline{\rho}^{#1}_{\bfC}\left(#2\right)}

\renewcommand{\vec}[1]{\underline{#1}}

\newcommand{\compGT}[4]{#2 {}^{#3}\!{\triangleleft}_{#1} #4}

\newcommand{\rap}[3]{#2 \star_{#1}{#3}}

\newcommand{\RMatchGT}[3]{\mathcal{M}^{{\text{\tiny $#1$}}}_{#2}(#3)}

 \newcommand{\inputTikz}[1]{%
  \ensuremath{\vcenter{\hbox{\includegraphics{diagrams/#1.pdf}}}}%
 }

\gdef\tdScale{0.5}
\newcommand{\inputTD}[2]{%
  \vcenter{\hbox{\includegraphics[page=#1, scale=#2]{diagrams/tD.pdf}}}%
}

\newcommand{\inputD}[2]{%
  \vcenter{\hbox{\includegraphics[scale=#1]{images/#2}}}%
}

\colorlet{h1color}{blue!70!black} %
\colorlet{h2color}{orange!90!black} %
\colorlet{h3color}{blue!40!white} %
\colorlet{h4color}{green!40!black} %

\title{On Stochastic Rewriting and Combinatorics\\
via Rule-Algebraic Methods\thanks{The author would like to thank Paul-Andr\'{e} Melli\`{e}s and Noam Zeilberger for fruitful discussions and valuable feedback.}}
\author{Nicolas Behr
\institute{Université de Paris, CNRS, IRIF\\ F-75006, Paris, France}
\email{nicolas.behr@irif.fr}
}
\def\titlerunning{On Stochastic Rewriting and Combinatorics}
\def\authorrunning{N. Behr}
\begin{document}
\maketitle

\begin{abstract}
Building upon the rule-algebraic stochastic mechanics framework, we present new results on the relationship of stochastic rewriting systems described in terms of continuous-time Markov chains, their embedded discrete-time Markov chains and certain types of generating function expressions in combinatorics. We introduce a number of generating function techniques that permit a novel form of static analysis for rewriting systems based upon marginalizing distributions over the states of the rewriting systems via pattern-counting observables.
\end{abstract}

\section{Introduction}

An important aspect of the standard theory of continuous-time Markov chains (CTMCs)~\cite{norris} concerns the well-known fact that the CTMC semantics may be equivalently described via a pair of discrete-time Markov chains (DTMCs), where the so-called \emph{embedded DTMC} encodes the probabilities for each of the possible transitions, and with the second DTMC encoding the jump-times for the transitions. This feature permits to design algorithms for \emph{simulating} CTMCs, for instance in the form of Gillespie's stochastic simulation algorithms for chemical reaction systems~\cite{Gillespie1977}, but in particular also in several variations for the simulation of stochastic rewriting systems, such as via the \emph{KaSim} simulation engine of the Kappa platform~\cite{Boutillier:2018aa}. The main contribution of the present paper consists in uncovering a hitherto unknown intimate relationship between three types of \emph{moment generating functions} that are constructable from the data that specifies a stochastic rewriting system, and for a chosen set of \emph{pattern count observables}: those of the CTMC itself, those of the embedded DTMC, and those of the (weighted) combinatorial species generated by the rewriting rules.

\section{Prerequisite: the rule algebra framework}\label{sec:RA}

The methodology developed in the present paper relies heavily upon the mathematical formalism introduced in~\cite{nbSqPO2019,bdg2016,bp2018,bdg2019,BK2020,bp2019-ext}, yet due to space restrictions, we will only provide some notations and essential definitions here, inviting the interested readers to consult loc.\ cit.\ for the full technical details.

\subsection{DPO- and SqPO-type compositional rewriting semantics}

Throughout this paper, we will consider categorical rewriting theories over categories that satisfy the following sets of properties (with DPO- and SqPO-semantics to be introduced below)\footnote{We invite the readers to consult~\cite{behrRaSiR} or~\cite{BK2020} for compact accounts of the relevant technical definitions of $\cM$-adhesive categories, pullbacks, pushouts, pushout complements, final pullback complements and their respective properties.}:
\begin{assumption}[cf.\ \cite{BK2020}, As.~1]\label{as:main}
    $\bfC\equiv(\bfC,\cM)$ is a finitary $\cM$-adhesive category with $\cM$-initial object, $\cM$-effective unions and epi-$\cM$-factorization. In the setting of \emph{Sesqui-Pushout (SqPO) rewriting}, we assume in addition that all final pullback complements (FPCs) along composable pairs of $\cM$-morphisms exist, and that $\cM$-morphisms are stable under FPCs.
\end{assumption}

\begin{definition}[\cite{habel2009correctness,ehrig2014mathcal}]
\emph{Conditions} over objects $X\in \obj{\bfC}$ are defined recursively: $\ac{true}_X$ is a condition; for any $\cM$-morphism $(X\hookrightarrow Y)$ and for any condition $\ac{c}_Y$ over $Y$, $\exists(X\hookrightarrow Y,\ac{c}_Y)$ is a condition; $\neg \ac{c}_X$ and $\ac{c}_X^{(1)}\land \ac{c}_X^{(2)}$ are conditions if $\ac{c}_X,\ac{c}_X^{(1)},\ac{c}_X^{(2)}$ are conditions. \emph{Satisfaction} of a condition $\ac{c}_X$ by a $\cM$-morphism $(a:X\hookrightarrow Z)$, denoted $a\vDash\ac{c}_X$, is also defined recursively: $a$ satisfies $\ac{true}_X$; $a$ satisfies $\exists(X\hookrightarrow Y,\ac{c}_Y)$ if there exists a $\cM$-morphism $(b:Y\hookrightarrow Z)$ such that $b\vDash \ac{c}_Y$ and $a=b\circ(Y\hookleftarrow X)$; $a\vDash\neg\ac{c}_X$ if not $a\vDash\ac{c}_X$, and $a\vDash\ac{c}_X^{(1)}\land\ac{c}_X^{(2)}$ if $a\vDash\ac{c}_X^{(1)}$ and $a\vDash\ac{c}_X^{(2)}$. Two conditions $\ac{c}_X$ and $\ac{c}_X'$ are defined to be \emph{equivalent}, denoted $\ac{c}_X\equiv\ac{c}_X'$, if $a\vDash\ac{c}_X \Leftrightarrow a\vDash\ac{c}_X'$ for all $\cM$-morphisms $(a:X\hookrightarrow Y)$. By standard convention, $\exists(a)\equiv\exists(a,\ac{true})$ and $\forall(a,\ac{c}_Y)\equiv\neg\exists(a,\neg\ac{c}_Y)$. We denote the class of all conditions over objects of $\bfC$ by $\cond{\bfC}$.
\end{definition}

Throughout this work, we will be exclusively interested in the following notion of \emph{rewriting rules}:
\begin{definition}
Let $\LinEq{\bfC}$ denote the set\footnote{Note that while $\LinAc{\bfC}$ is typically a proper class (i.e.\ of $\cM$-linear rules with conditions), we make the tacit assumption here that $\LinEq{\bfC}$ forms a set. We opt throughout this paper for a ``right-to-left'' convention for rules and their compositions that is non-standard w.r.t.\ the standard graph rewriting literature, yet which is preferable in the setting of rule-algebraic computations. This is due to rules $r$ giving rise to linear operators $\rho(\delta(r))$ (Definition~\ref{def:RAcanReps}) which by standard mathematical convention left-compose, whence our notational convention is ultimately motivated by the representation property stated in Theorem~\ref{thm:CanRepProps}.} of equivalence classes of \emph{$\cM$-linear rules} $\LinAc{\bfC}$~\cite{behrRaSiR}, 
\begin{equation}
\LinEq{\bfC}:=
\{R=(r=(O\xleftarrow{o}K\xrightarrow{i}I),\ac{c}_I)\mid o,i\in \cM\,,\;\ac{c}_I\in \cond{\bfC}\}\diagup\sim\,,
\end{equation}
where $R\sim R'$ iff there exist isomorphisms $(\omega:O\rightarrow O')$, $(\kappa:K\rightarrow K')$ and $(\iota:I\rightarrow I')$ such that $\omega\circ o=o'\circ \kappa$ and $\iota\circ i=i'\circ \kappa$, and if in addition $\ac{c}_I\equiv\ac{c}_{I}'$ (i.e.\ if $\forall (m:I\hookrightarrow X)\in \cM: m\vDash \ac{c}_I\ \Leftrightarrow m\vDash\ac{c}_{I}'$).
\end{definition}

The two key definitions of rewriting theory for the formulation of rule algebras (cf.\ Section~\ref{sec:RAdef}) are the definition of the \emph{action of rules on objects} and of the \emph{sequential composition of rules}.
\begin{definition}[Direct derivations; cf.~\cite{BK2020}, Def.~3]
    Let $r=(O\hookleftarrow K\hookrightarrow I)\in\Lin{\bfC}$ and $\ac{c}_I\in\cond{\bfC}$ be concrete representatives of some equivalence class $R=[(r,\ac{c}_I)]_{\sim}\in\LinEq{\bfC}$, and let $X,Y\in\obj{\bfC}$ be objects. Then a \emph{type $\bT$ direct derivation} is defined as a commutative diagram such as below right, where all morphism are in $\cM$ (and with the left representation a shorthand notation)
\begin{equation}\label{eq:DD}
\inputD{1}{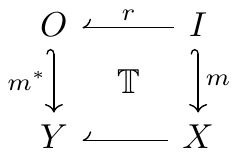}\quad :=\quad 
\inputD{1}{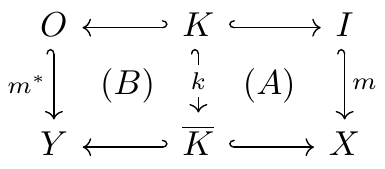}\,.
\end{equation}
with the following pieces of information required relative to the type:
\begin{enumerate}
    \item $\mathbf{\bT=DPO}$: given $(m:I\hookrightarrow X)\in\cM$, $m$ is a \emph{DPO-admissible match of $R$ into $X$}, denoted $m\in\MatchGT{DPO}{R}{X}$, if $m\vDash \ac{c}_I$ and $(A)$ is constructable as a \emph{pushout complement}, in which case $(B)$ is constructed as a \emph{pushout}.
    \item $\mathbf{\bT=SqPO}$: given $(m:I\hookrightarrow X)\in\cM$, $m$ is a \emph{SqPO-admissible match of $R$ into $X$}, denoted $m\in\MatchGT{SqPO}{R}{X}$, if $m\vDash\ac{c}_I$, in which case $(A)$ is constructed as a \emph{final pullback complement} and $(B)$ as a \emph{pushout}.
    \item $\mathbf{\bT=DPO^{\dag}}$: given just the ``plain rule'' $r$ and $(m^{*}:O\hookrightarrow Y)\in\cM$, $m^{*}$ is a \emph{DPO${}^{\dag}$-admissible match of $r$ into $X$}, denoted $m\in\MatchGT{DPO^{\dag}}{r}{Y}$, if $(B)$ is constructable as a \emph{pushout complement}, in which case $(B)$ is constructed as a \emph{pushout}.
\end{enumerate}
For types $\bT\in\{DPO,SqPO\}$, we will sometimes employ the notation $R_m(X)$ for the object $Y$.
\end{definition}

\begin{definition}\label{def:rRel}
For rewriting of type $\bT\in \{DPO,SqPO\}$, $X_0\in \obj{\bfC}_{\cong}$ and a set of rules $\cR=\{R_j\}_{j=1}^n$, we recursively define a family of \emph{reachability relations} $\{\Rightarrow_{(i)}\}_{i\geq 0}$ on $\obj{\bfC}_{\cong}^{\times\:2}$ as
\begin{equation}
\begin{aligned}
X\Rightarrow_{(1)}Y\; 
&\text{iff } \exists R\in\cR\,, m\in \cM^{\bT}_{R}(X):Y\cong R_m(X)\\
\forall n\geq 1:\quad X\Rightarrow_{(i+1)}Y
\; &\text{iff } \exists Z\in \obj{\bfC}_{\cong}\,, R\in\cR\,, m\in \cM^{\bT}_{R}(Z):Y\cong R_m(Z) \land X\Rightarrow_{(i)}Z\,.
\end{aligned}
\end{equation}

\end{definition}

\begin{definition}[Rule compositions; cf.~\cite{BK2020}, Def.~4]\label{def:Rcomp}
 Let $R_1,R_2\in \LinEq{\bfC}$ be two equivalence classes of rules with conditions, and let $r_j\in\Lin{\bfC}$ and $\ac{c}_{I_j}$ be representatives of $R_j$ (for $j=1,2$). For $\bT\in\{DPO,SqPO\}$, a span $\mu=(I_2\hookleftarrow M_{21}\hookrightarrow O_1)$ of $\cM$-morphisms is a \emph{$\bT$-admissible match of $R_2$ into $R_1$}, denoted $\mu\in \MatchGT{\bT}{R_2}{R_1}$,  if the diagram below is constructable (with $N_{21}$ constructed by taking pushout)
 \begin{equation}\label{eq:defRcomp}
 \inputD{1}{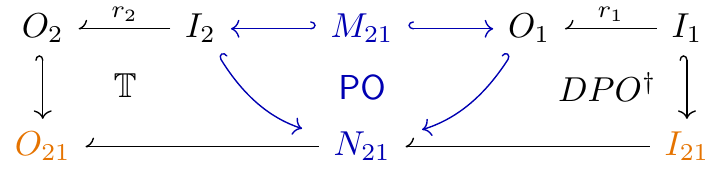}
\end{equation}
and if $\ac{c}_{I_{21}}\not{\!\!\dot{\equiv}}\,\,\ac{false}$. Here, the condition $\ac{c}_{I_{21}}$ is computed as (cf.\ \cite{BK2020} for the definitions of $\Shift$ and $\Trans$)
\begin{equation}
    \ac{c}_{I_{21}}:=\Shift(I_1\hookrightarrow I_{21},\ac{c}_{I_1})\;\land\; \Trans(N_{21}\leftharpoonup I_{21},\Shift(I_2\hookrightarrow N_{21},\ac{c}_{I_2}))\,.
\end{equation}
In this case, we define the \emph{type $\bT$ composition of $R_2$ with $R_1$ along $\mu$}, denoted $\compGT{\bT}{R_2}{\mu}{R_1}$, as
\begin{equation}
\compGT{\bT}{R_2}{\mu}{R_1}:=[(O_{21}\leftharpoonup I_{21};\ac{c}_{I_{21}})]_{\sim}\,,
\end{equation}
where $(O_{21}\leftharpoonup I_{21}):=(O_{21}\leftharpoonup N_{21})\circ(N_{21}\leftharpoonup I_{21})$ (with $\circ$ the \emph{span composition} operation).
\end{definition}

\subsection{DPO- and SqPO-type rule algebra frameworks}\label{sec:RAdef}

Referring yet again to~\cite{behrRaSiR,BK2020} for the full technical details, suffice it here to quote the essential definitions:

\begin{definition}[\textbf{Rule algebras}; \cite{BK2020}, Def.~5]
    Let $\bT\in\{DPO,SqPO\}$ be the rewriting type, and let $\bfC$ be a category satisfying the relevant variant of Assumption~\ref{as:main}. Let $\overline{\cR}_{\bfC}$ be an $\bR$-vector space, defined via a bijection $\delta:\LinEq{\bfC}\xrightarrow{\cong}\mathsf{basis}(\overline{\cR}_{\bfC})$ from the set of equivalence classes of linear rules with conditions to the set of basis vectors of $\overline{\cR}_{\bfC}$. Let $\rap{\bT}{}{}$ denote the \emph{type $\bT$ rule algebra product}, a binary operation defined  via its action on basis elements $\delta(R_1),\delta(R_1)\in \overline{\cR}_{\bfC}$ (for $R_1,R_2\in \LinEq{\bfC}$) as
    \begin{equation}
        \rap{\bT}{\delta(R_2)}{\delta(R_1)}:=\sum_{\mu\in \RMatchGT{\bT}{R_2}{R_1}}\delta\left(\compGT{\bT}{R_2}{\mu}{R_1}\right)\,.
    \end{equation}
    We refer to $\overline{\cR}_{\bfC}^{\bT}:=(\overline{\cR}_{\bfC},\rap{\bT}{}{})$ as the \emph{$\bT$-type rule algebra over $\bfC$}.
\end{definition}

\begin{definition}[Representations; \cite{BK2020}, Def.~6]\label{def:RAcanReps}
    Let $\hat{\bfC}$ be defined as the \emph{$\bR$-vector space} whose set of basis vectors is isomorphic to the set of iso-classes of objects of $\bfC$ via a bijection $\ket{.}:\obj{\bfC}_{\cong}\rightarrow \mathsf{basis}(\hat{\bfC})$. Then the \emph{$\bT$-type canonical representation of the $\bT$-type rule algebra over $\bfC$}, denoted $\overline{\rho}^{\bT}_{\bfC}$, is defined as the morphism $\overline{\rho}^{\bT}_{\bfC}:\overline{\cR}_{\bfC}^{\bT}\rightarrow End_{\bR}(\hat{\bfC})$ specified via 
    \begin{equation}
    \forall R\in\LinEq{\bfC},X\in\obj{\bfC}_{\cong}:
    \quad \canRep{\bT}{\delta(R)}\ket{X}:=\sum_{m\in\MatchGT{\bT}{R}{X}}\ket{R_m(X)}\,.
    \end{equation}
\end{definition}

The salient part of the rule-algebra framework in view of computational techniques is the following result on the interaction of the rule algebra product and the notion of representations:
\begin{theorem}[\cite{BK2020}, Thm.~3]\label{thm:CanRepProps}
    $\overline{\rho}^{\bT}_{\bfC}$ is an \emph{algebra homomorphism}, i.e.\ (for $\rho\equiv\rho_{\bfC}^{\bT}$)
    \begin{equation}
    (i)\quad \rho(\delta(R_{\mIO}))=Id_{End(\hat{\bfC})}\qquad 
    (ii)\quad \rho(\delta(R_2))\rho(\delta(R_1))=\rho\left(\rap{\bT}{\delta(R_2)}{\delta(R_1)}\right)\,.
    \end{equation}
\end{theorem}

\section{Stochastic rewriting systems}\label{sec:SM}

For either DPO- or SqPO-type rewriting semantics over a category $\bfC$ satisfying Assumption~\ref{as:main} a \emph{stochastic rewriting system (SRS)} is specified via providing the following pieces of data:
\begin{enumerate}
\item An \emph{input state}, i.e.\ a \emph{probability distribution} indexed by isomorphism classes of objects of $\bfC$:
\begin{equation}
	\ket{\Psi_0}=\sum_{X\in \obj{\bfC}_{\cong}} p_X(0) \ket{X}
\end{equation}
\item A \emph{set of transitions} $\{(\kappa_j,R_j)\}_{j=1}^n$, whence a set of (finitely many) pairs of \emph{base rates} $\kappa_j\in \bR_{>0}$ and \emph{linear rules with conditions} $R_j\in \LinAc{\bfC}$, with $R_j=(r_j,\ac{c}_{I_j})$.
\end{enumerate}

\begin{example} Consider the case where $\bfC=\mathbf{FinSet}$ (the finitary restriction of the category of sets and set functions), with an input state a ``pure state'' $\ket{\Psi_0}=\ket{\mIO}$ (i.e.\ a probability distribution with probability one for the state associated to the empty set), and with a transition set given by $\{(\kappa_{+},R_{+}),(\kappa_{-}, R_{-})\}$, where $R_{+}:=(\bullet\leftarrow \mIO\rightarrow\mIO, \ac{c}_{+})$ and $R_{-}:=(\mIO\leftarrow \mIO\rightarrow \bullet,\ac{c}_{-})$. If $\ac{c}_{+}=\ac{c}_{-}=\ac{true}$, we reproduce precisely a classical \emph{birth-death process} (with initial state the state with no ``particles''). However, we may as well choose non-trivial conditions, for example $\ac{c}_{+}^{(3)}=\exists(\mIO\hookrightarrow \bullet\;\bullet\;\bullet)$, in which case we would obtain a system in which once the number of vertices falls below $3$ vertices, the ``death'' transition will with probability $1$ eventually delete all remaining vertices.
\end{example}

Back to the general stochastic mechanics framework, we assemble from the input data the definition of a CTMC as follows (with $\rho\equiv \rho_{\bfC}^{\bT}$, $\bT\in\{DPO,SqPO\}$ and $\hat{\bO}_{\bT}$ as in~\eqref{eq:JCgen} of Section~\ref{sec:obs} below):
\begin{equation}\label{eq:defCTMC}
\begin{aligned}
\tfrac{d}{dt}\ket{\Psi(t)}&= H\ket{\Psi(t)}\,,\quad \ket{\Psi(0)}=\ket{\Psi_0}\\
H&:= \hat{H}+H_{\bO}\,,\quad \hat{H}=\rho(h)\,,\; H_{\bO}=-\hat{\bO}_{\bT}(h)\,,\quad
h=\sum_{j=1}^n \kappa_j\rho(\delta(R_j))\,.
\end{aligned}
\end{equation}
Compared to the traditional CTMC literature~\cite{norris}, note that $\hat{H}$ is a linear operator with \emph{off-diagonal} non-negative entries only, while $H_{\bO}$ is a diagonal linear operator with non-positive entries. In order to demonstrate that $H$ indeed qualifies as a so-called \emph{infinitesimal stochastic operator} (also referred to as conservative stable $Q$-matrix), note that by virtue of the jump-closure property (cf.\ \eqref{eq:JCgen} in Section~\ref{sec:obs}), 
\begin{equation}\label{eq:Hprop}
\bra{}H = \bra{}(\hat{H}+H_{\bO})=\bra{}(\rho(h)-\bO(h)) = 0_{\bR}\,.
\end{equation}
In other words, the above equation expresses that each row of $H$ sums to zero, which together with the (non-) negativity of the off- and on-diagonal entries confirms that $H$ is a valid infinitesimal CTMC generator.

\subsection{Pattern count observables}\label{sec:obs}

Unlike in the setting of chemical reaction systems or, equivalently, of rewriting systems over vertex-only graphs, where certain combinatorial techniques exist to directly compute the full solution for $\ket{\Psi(t)}$ (cf.\ e.g.\ \cite{bdg2019}), as it stands the definition of a CTMC for a stochastic rewriting system provided in~\eqref{eq:defCTMC} does not lend itself easily in its own right for practical computations. Intuitively, for a stochastic rewriting system such as the one we will study in detail in Section~\ref{sec:PRBTs} evolving over planar rooted binary trees (PRBTs), the set of isomorphism classes reachable even in a very small number of transitions and given a ``pure'' input state $\ket{\Psi_0}=\ket{X_0}$ (for some iso-class $X_0$) can be of an enormous size and complexity. For instance, there are already $\tfrac{200!}{100!}\approx 10^{217}$ isomorphism classes of PRBTs reachable in just $100$ iterations of the R\'{e}my generator studied in Section~\ref{sec:PRBTs}. Since CTMC theory describes the time-evolution of the probability distribution $\ket{\Psi(t)}$ over \emph{all} states reachable by the transitions from the input state, it is in practice very often entirely infeasible to aim for a direct calculation of $\ket{\Psi(t)}$ itself.

The remedy for the aforementioned conceptual problem is the introduction of the concept of \emph{pattern count observables}. Referring to~\cite{BK2020} for the precise derivation, depending on the rewriting semantics $\bT\in\{DPO,SqPO\}$ utilized, let us introduce the following definitions:
\begin{equation}\label{eq:defSqPOjc}
\hat{O}_{P,q;\ac{c}_P}
:= \rho_{\bfC}^{DPO}\left(\delta\left(P\xleftarrow{q}Q\xrightarrow{q}P, \ac{c}_P\right)\right)\,,\qquad
\hat{O}_{P;\ac{c}_P}:= \rho_{\bfC}^{SqPO}\left(\delta\left(P\xleftarrow{id}P\xrightarrow{id}P, \ac{c}_P\right)\right)
\end{equation}
To better understand the meaning of the above definitions, it is important to note the so-called \emph{jump-closure properties} of DPO- and SqPO-types, respectively (cf.\ \cite[Thm.~4]{BK2020}):
\begin{equation}\label{eq:JCgen}
\begin{aligned}
\forall R=(O\xleftarrow{o}{\color{h2color}K\xrightarrow{i}}{\color{h1color}I},{\color{h4color}\ac{c}_I})\in \LinAc{\bfC}:\quad 
\bra{}\rho_{\bfC}^{\bT}(\delta(R))
&=\bra{}\hat{\bO}_{\bT}(\delta(R))\\
\hat{\bO}_{DPO}(\delta(R))&:= \hat{O}_{{\color{h1color}I},{\color{h2color}i};{\color{h4color}\ac{c}_I}}\,,\quad 
\hat{\bO}_{SqPO}(\delta(R)):= \hat{O}_{{\color{h1color}I};{\color{h4color}\ac{c}_I}}\,. 
\end{aligned}
\end{equation}
In other words, $\hat{O}_{{\color{h1color}I},{\color{h2color}i};{\color{h4color}\ac{c}_I}}$ and $\hat{O}_{{\color{h1color}I};{\color{h4color}\ac{c}_I}}$ permit to \emph{count} the number of matches of the rule $R=(O\xleftarrow{o}{\color{h2color}K\xrightarrow{i}}{\color{h1color}I},{\color{h4color}\ac{c}_I})$ in DPO- and SqPO-semantics, respectively. More explicitly, we find that
\begin{equation}\label{eq:jcExplicit}
\begin{aligned}
\bra{}\hat{O}_{{\color{h1color}I},{\color{h2color}k};{\color{h4color}\ac{c}_I}}\ket{X} &= 
\vert \cM_{({\color{h1color}I}{\color{h2color}\xleftarrow{i}K\xrightarrow{i}}{\color{h1color}I},{\color{h4color}\ac{c}_I})}^{DPO}(X)\vert
&\!\!\!\!\!\!\!\!&= 
\vert \cM_{(O\xleftarrow{o}{\color{h2color}K\xrightarrow{i}}{\color{h1color}I},{\color{h4color}\ac{c}_I})}^{DPO}(X)\vert\\
\bra{}\hat{O}_{{\color{h1color}I};{\color{h4color}\ac{c}_I}}\ket{X} &= 
\vert \cM_{({\color{h1color}I}\xleftarrow{id}{\color{h1color}I}\xrightarrow{id}{\color{h1color}I},{\color{h4color}\ac{c}_I})}^{SqPO}(X)\vert
&\!\!\!\!\!\!\!\!&=\vert \cM_{(O\xleftarrow{o}{\color{h2color}K\xrightarrow{i}}{\color{h1color}I},{\color{h4color}\ac{c}_I})}^{SqPO}(X)\vert\,.
\end{aligned}
\end{equation}

\begin{example}
The simplest type of observables encountered in practice are the ``plain'' pattern-counting observables\footnote{Here, $=$ is the equality of linear operators, whereby two linear operators are equal if they agree in every matrix entry.} %
$\hat{O}_{\color{h1color}P}=\hat{O}_{{\color{h1color}P},id_{\color{h1color}P};\ac{true}}=\hat{O}_{{\color{h1color}P};\ac{true}}$, with typical examples including $\hat{O}_{\bullet}$ (counting \emph{vertices}), $\hat{O}_{\bullet\;\bullet}$ (counting \emph{pairs of vertices}) and $\hat{O}_{\bullet\!\!-\!\!\bullet}$ (counting \emph{(undirected) edges}). In contrast, for example in DPO rewriting the variant %
$\hat{O}_{{\color{h1color}\bullet},{\color{h2color}\mIO\hookrightarrow \bullet}; {\color{h4color}\ac{true}}}$ effectively counts vertices that are not linked to any other vertices via incident edges, whence by~\eqref{eq:jcExplicit} in particular %
$\hat{O}_{{\color{h1color}\bullet},{\color{h2color}\mIO\hookrightarrow \bullet};{\color{h4color}\ac{true}}} \neq \hat{O}_{{\color{h1color}\bullet};{\color{h4color}\ac{true}}}$, while in both DPO- and SqPO-rewriting,
\[
\hat{O}_{{\color{h1color}\bullet\;\bullet},{\color{h2color}\bullet\;\bullet\hookrightarrow \bullet\;\bullet};{\color{h4color}\neg\exists(\bullet\;\bullet\hookrightarrow \bullet\!\!-\!\!\bullet)}}
=\hat{O}_{{\color{h1color}\bullet\;\bullet};{\color{h4color}\neg\exists(\bullet\;\bullet\hookrightarrow \bullet\!\!-\!\!\bullet)}}
\]
yields a linear operator that in effect counts pairs of vertices not linked by an edge.
\end{example}

\subsubsection{Exponential moment-generating functions}

As advocated in~\cite{bdg2016,bp2019-ext,nbSqPO2019,bdg2019,BK2020}, observables are a crucial concept in rule-algebraic rewriting theory, since it will often be the case that one considers distributions over isomorphism classes of objects in the given category, which due to the typically enormous number of such classes cannot be computed explicitly, or would even be insensible to compute. Instead, choosing a set of pattern-counting observables permits to extract \emph{partial information} from the distributions, which may then provide important insights on the behavior of the rewriting system. 

\begin{definition}\label{def:EMGF}
Given a stochastic rewriting system as defined in~\eqref{eq:defCTMC}, and for a finite set of observables $\{\hat{O}_1,\dotsc\hat{O}_m\}$, the \emph{exponential moment-generating function (EMGF)} $\cM(t;\vec{\omega})$ is defined as
\begin{equation}\label{eq:defEMGF}
\cM(t;\vec{\omega}):=\bra{}e^{\vec{\omega}\cdot\vec{\hat{O}}}\ket{\Psi(t)}\,.
\end{equation} 
Here, we employed the shorthand notation $\vec{\omega}\cdot\vec{\hat{O}}:=\sum_{j=1}^m\omega_j\hat{O}_j$, with $\omega_1,\dotsc,\omega_m$ \emph{formal variables}.
\end{definition}

The interpretation of $\cM(t;\vec{\omega})$ is that it serves as a form of bookkeeping device for all (joint) moments of the chosen observables, in the sense that one may expand~\eqref{eq:defEMGF} into a formal power series of the form
\begin{equation}\label{eq:EMGFcoeffs}
\cM(t;\vec{\omega})=\sum_{\vec{k}\geq \vec{0}}\tfrac{\vec{\omega}^{\vec{k}}}{\vec{k!}}\bra{}\vec{\hat{O}}^{\vec{k}}\ket{\Psi(t)}\,,\quad \vec{x}^{\vec{k}}\equiv x_1^{k_1}\cdots x_m^{k_m}\,.
\end{equation}
As a caveat, one should note that it depends strongly on the chosen rewriting system whether or not the formula for $\cM(t;\vec{\omega})$ is mathematically well-posed, in the sense that all moments exist and are finite. Since currently no general theory for providing conditions on the rewriting system that ensure these properties, one must resort in practice to either simulation data or to other empirical methods in order to (at least approximately) verify the well-posedness. Setting these technical issues aside, the main motivation for studying exponential moment-generating functions is the fact that if indeed all moments exist and are finite, there exists a Legendre transform from $\cM(t;\vec{\omega})$ to the \emph{probability-generating function} $P(t;\vec{x})$ (i.e.\ a change of variables $\omega_j\to \ln x_j$ for $j=1,\dotsc,m$):
\begin{equation}\label{eq:LGT}
P(t;\vec{x}):=\cM(t;\underline{\ln x})\,.
\end{equation}

\begin{lemma}
The formal power series $P(t;\vec{x})$ is the \emph{ordinary probability-generating function} for the counts $n_1,\dotsc,n_m$ of the observables $\hat{O}_1,\dotsc,\hat{O}_m$, which entails that with $n_j(X):=\bra{}\hat{O}_j\ket{X}$ (for $j=1,\dotsc,m$),
\begin{equation}\label{eq:PGF}
P(t;\vec{x})=\sum_{\vec{n}\geq \vec{0}}\vec{x}^{\vec{n}} p_{\vec{n}}(t)\,,\quad
p_{\vec{n}}(t):=Pr(\{X\in \obj{\bfC}_{\cong}\mid \vec{n(X)}=\vec{n} \}\;\vert \text{at time $t$})\,.
\end{equation}
\end{lemma}
\begin{proof}
Starting from the definition of $\cM(t;\vec{\omega})$ and~\eqref{eq:LGT},
\begin{equation}
P(t;\vec{x})=\cM(t;\vec{\ln x})=\bra{}\vec{x}^{\vec{\hat{O}}}\ket{\Psi(t)}
= \sum_{X\in \obj{\bfC}_{\cong}}p_X(t)\bra{}\vec{x}^{\vec{\hat{O}}}\ket{X}
= \sum_{X\in \obj{\bfC}_{\cong}}p_X(t)\vec{x}^{\vec{n(X)}}\braket{}{X}\,.
\end{equation}
Since by definition $\braket{}{X}=1_{\bR}$, this is almost of the form in~\eqref{eq:PGF}, except that the summation is not over the values of observable counts, but instead over the possible states (i.e.\ isomorphism classes of objects). However, as is a well-known technique in the combinatorics literature~\cite{FlajoletSedgewick}, since $\ket{\Psi(t)}$ is a probability distribution over a \emph{countable} space of states (i.e.\ of isomorphism classes of objects of $\bfC$), the remaining argument is a simple re-partition of the probability distribution of the form
\begin{equation}
\sum_{X\in \obj{\bfC}_{\cong}}p_X(t)\vec{x}^{\vec{n(X)}}
=\sum_{\vec{n}\geq \vec{0}}\;\big(
\overbrace{
\sum_{\substack{X\in \obj{\bfC}_{\cong}\\ \vec{n(X)}=\vec{n}}} p_X(t)}^{=: p_{\vec{n}}(t)}\big)\; \vec{x}^{\vec{n}}\,.
\end{equation}
\end{proof}

To summarize this part of the general framework, let us emphasize that the passage from exponential moment-generating functions to probability-generating functions in principle permits to extract some very detailed information from stochastic rewriting systems, with some worked examples of this kind to be found in~\cite{bp2019-ext,nbSqPO2019,bdg2019,BK2020}. Modulo the question of how to practically compute either the exponential moment- or the probability-generating functions, we have thus obtained a method to project the typically extremely large state-spaces to the more tractable subspaces in which the state-space is partitioned into subspaces determined by the values of the counts of a given set of pattern observables.

\subsection{EMGF evolution equations}

The concept of exponential moment-generating functions (EMGFs) for stochastic rewriting systems opens the possibility for a form of \emph{static analysis technique} that aims at computing the time-dependent EMGF from an \emph{evolution equation}, based upon the following key result:
\begin{theorem}[\cite{bdg2019}, Thm.~3.24]
For a CTMC such as in~\eqref{eq:defCTMC}, an EMGF $\cM(t;\vec{\omega})$ as defined in~\eqref{eq:defEMGF} satisfies the \emph{formal evolution equation}
\begin{equation}\label{eq:EMGFfe}
\tfrac{\partial}{\partial t}\cM(t;\vec{\omega})= \sum_{q\geq 1}\tfrac{1}{q!}\bra{}\left(
ad_{\vec{\omega}\cdot\vec{\hat{O}}}^{\:\circ\:q}(\hat{H})
\right)e^{\vec{\omega}\cdot\vec{\hat{O}}}\ket{\Psi(t)}\,,
\end{equation}
where the \emph{adjoint action} $ad$ is recursively defined for any two endomorphisms $A,B\in End(\hat{\bfC})$ as
\begin{equation}
ad_{A}^{\:\circ\:0}(B):= A\,,\; 
\forall q>0:\; ad_{A}^{\:\circ\:(q+1)}(B)=
[A,ad_{A}^{\:\circ\:q}(B) ]\,.
\end{equation}
Here, $[A,B]:=AB-BA$ is referred to in the mathematics literature as the \emph{commutator} of $A$ and $B$.
\end{theorem}

As it stands, the formal evolution equation in~\eqref{eq:EMGFfe} is not of any computational value, since a priori it is not possible to evaluate concretely the infinite series of coefficients. However, there exists a criterion whose satisfaction permits to proceed further:
\begin{definition}[\cite{bdg2019}, Def.~4.1]
A set of observables $\hat{O}_1,\dotsc,\hat{O}_m$ (for finite $m$) is \emph{polynomially jump-closed} with respect to a CTMC as in~\eqref{eq:defCTMC} if the following conditions hold true:
\begin{equation}\label{eq:PJC}
(\mathsf{PJC})\quad \forall q\in \bZ_{>0}:\exists \vec{N(n)}\in \bZ_{\geq0}^m,\gamma_q(\vec{\omega},\vec{k})\in \bR:\; 
\bra{}ad_{\vec{\omega}\cdot\vec{\hat{O}}}^{\:\circ\:q}(\hat{H}) 
=\sum_{\vec{k}=\vec{0}}^{\vec{N(q)}}\gamma_{\vec{k}}(\vec{\omega},\vec{k})\bra{}\vec{\hat{O}}^{\vec{k}}\,.
\end{equation}
\end{definition}

Polynomial jump-closure has the benefit of permitting to transform the formal evolution equation of~\eqref{eq:EMGFfe} into a computationally much more accessible evolution equation on formal power series, a process referred to in~\cite{bdg2019} as \emph{combinatorial conversion}:
\begin{theorem}[\cite{bdg2019}, Thm.~4.2]\label{thm:CC}
Let $\hat{O}_1,\dotsc,\hat{O}_m$ be a (finite) set of pattern observables that is polynomially jump-closed with respect to a CTMC as in~\eqref{eq:defCTMC}. Then the EMGF $\cM(t;\vec{\omega})$ satisfies the \emph{evolution equation}
\begin{equation}\label{eq:CCthm}
\tfrac{\partial}{\partial t}\cM(t;\vec{\omega})=\bD(\vec{\omega},\vec{\partial\omega}) \cM(t;\vec{\omega})\,,\quad \bD(\vec{\omega},\vec{\partial\omega}):=
\left.\left(
\bra{} e^{ad_{\vec{\omega}\cdot\vec{\hat{O}}}}(H)
\right)\right\vert_{\vec{\hat{O}}\mapsto \vec{\partial\omega}}
\end{equation}
\end{theorem}

Referring to~\cite{bdg2019,BK2020} for a number of applications of the above theorem, suffice it here for brevity to mention that the theorem in favorable cases permits to establish a direct contact between techniques introduced in applied research fields such as network, physical and social sciences on the one hand, and rewriting-theoretic methods on the other hand.

\section{New horizons: Weighted combinatorial species and embedded DTMCs}\label{sec:WCSsandDTMCs}

We envision that the rich theory of analytical combinatorics~\cite{FlajoletSedgewick} with its numerous computational methods for reasoning about pattern counts in combinatorial structures and in particular about their asymptotic behavior in the limit of large structure sizes could harbor an enormous potential also for reasoning about stochastic rewriting systems. For instance, certain aspects of combinatorial theory parallel surprisingly closely the notions of $\cM$-adhesive categories, concretely the various operations that permit to obtain valid specifications of combinatorial species from simpler species, which strikingly resemble the operations listed in the work of Ehrig~\cite{ehrig:2006aa} for generating $\cM$-adhesive categories from $\cM$-adhesive categories (e.g.\ the operations of sum, product, functor and comma categories). Yet to date, this clear correspondence is not straightforward to exploit, in part since many combinatorial specifications of graph-like structures require in effect the theory of $\cM$-adhesive categories combined with structural constraints in order to encode them rewriting-theoretically, but more importantly since generating-function techniques are not in the least common-place in rewriting theory. As a modest first step towards a rewriting-based variant of combinatorial theories, we introduce here a notion of multi-variate generating functions that permits to much more clearly compare the combinatorial theory with rewriting theory. In Section~\ref{sec:PRBTs}, we will moreover present a complete worked example in order to illustrate this new viewpoint.

Consider a rewriting system over some suitable category $\bfC$ that consists of a finite set of rules with conditions $R_1,\dotsc,R_n\in \LinAc{\bfC}$. For some choice of parameters $\gamma_1,\dotsc,\gamma_n\in \bR$, define a linear operator
\begin{equation}\label{eq:defG}
\hat{G}:=\sum_{j=1}^n \gamma_j\rho(\delta(R_j))\,.
\end{equation} 
The readers will immediately notice that $\hat{G}$ has a natural interpretation as a linear operator that encodes ``application of the rules $R_1,\dotsc,R_n$ in all possible ways, and weighted by the parameters $\gamma_1,\dotsc,\gamma_n$''. In close analogy to the definition of exponential moment-generating functions (EMGFs) in the previous section, if we in addition choose a (finite) set of pattern observables $\hat{O}_1,\dotsc,\hat{O}_m$, we may define what is conventionally also referred to as an EMGF (albeit this time not for moments of a probability distribution, but of an arbitrary distribution):
\begin{equation}
\cG(\lambda;\vec{\omega}):=\bra{}e^{\vec{\omega}\cdot\vec{\hat{O}}}e^{\lambda\hat{G}}\ket{X_0}
\end{equation}
Here, we chose an ``initial state'' $\ket{X_0}\in \hat{\bfC}$, yet we could have in principle equally well chosen some arbitrary ``initial distribution'' $\ket{\Phi(0)}$ (possibly subject to suitable summability conditions). Equipped with the insights from the previous section, it is straightforward to develop the analogue of the \emph{formal EMGF evolution equation} for $\cG(\lambda;\vec{\omega})$:
\begin{equation}\label{eq:EMGFfeG}
\tfrac{\partial}{\partial \lambda}\cG(\lambda;\vec{\omega})
=\bra{}\left(
e^{ad_{\vec{\omega}\cdot\vec{\hat{O}}}}\hat{G}
\right)e^{\vec{\omega}\cdot\vec{\hat{O}}}e^{\lambda\hat{G}}\ket{X_0}
\end{equation}
Applying the version of the jump-closure theorem appropriate for the chosen rewriting semantics (DPO or SqPO), the above formal evolution equation may be converted into a proper evolution equation on formal power series if  the following modified version of polynomial jump-closure holds:
\begin{equation}\label{eq:defPJC2}
(\mathsf{PJC}')\quad \forall q\in \bZ_{{\color{blue}\geq}0}:\exists \vec{N(n)}\in \bZ_{\geq0}^m,\gamma_q(\vec{\omega},\vec{k})\in \bR:\; 
\bra{}ad_{\vec{\omega}\cdot\vec{\hat{O}}}^{\:\circ\:q}(\hat{G}) 
=\sum_{\vec{k}=\vec{0}}^{\vec{N(q)}}\gamma_{\vec{k}}(\vec{\omega},\vec{k})\bra{}\vec{\hat{O}}^{\vec{k}}
\end{equation}
If a given set of observables satisfies $(\mathsf{PJC}')$, the formal evolution equation~\eqref{eq:PJC} for the EMGF $\cG(\lambda;\vec{\omega})$ may be refined into 
\begin{equation}
\tfrac{\partial}{\partial\lambda}\cG(\lambda;\vec{\omega}) = \bG(\vec{\omega},\vec{\partial\omega})\cG(\lambda;\vec{\omega})\,,\quad 
\bG(\vec{\omega},\vec{\partial \omega})=\left.\left(
\bra{} e^{ad_{\vec{\omega}\cdot\vec{\hat{O}}}}(\hat{G})
\right)\right\vert_{\vec{\hat{O}}\mapsto\vec{\partial \omega}}\,.
\end{equation}

\subsection{Weighted combinatorial species}\label{sec:wcs}

In modern standard approaches to combinatorics such as described in the seminal book~\cite{FlajoletSedgewick}, the central technical tool consists in the theory of \emph{combinatorial species}. According to the work of A.~Joyal~\cite{Joyal1981}, the specification of a combinatorial species (also referred to as an $F$-structure) consists in providing a functor $F:\textbf{B}\rightarrow\textbf{B}$, where $\mathbf{B}:=\mathbf{FinSet}^{*}$ denotes the groupoid of finite sets and bijections\footnote{Recall that a groupoid is a category in which all morphisms are isomorphism; consequently, given an arbitrary category $\mathbf{C}$ (such as e.g.\ $\mathbf{FinSet}$, the category of finite sets and total functions), one may obtain a groupoid $\mathbf{C}^{*}$ called the \emph{core of $\bfC$} via restricting the morphisms of $\bfC$ to isomorphisms.} $\bfB=\mathbf{FinSet}^{*}$. The functoriality of $F$ entails in particular that if $F[A]$ is an $F$-structure for a given set of ``labels'' $A$, and if $f:A\rightarrow B$ is an isomorphism of finite sets (i.e.\ a morphism of $\bfB$), then one obtains a morphism of $F$-structures $F[f]:F[A]\rightarrow F[B]$ that is also an isomorphism (since functors preserve isomorphisms). On the other hand, two finite sets $A$ and $B$ can only be isomorphic if they have the same number of elements, $|A|=|B|$, thus we recover the more intuitive interpretation of $F$-structures as ``structures induced by $F$ and invariant under relabeling''. We posit that in situations where the combinatorial species is of a certain ``graph-like nature'' (to be specified in further details momentarily), it may be advantageous to consider instead of classical species theory an alternative approach that is based upon category-theoretical structure of the combinatorial structures other than the defining species structure. More precisely, in the case where the species at hand is either giving rise to an \emph{(finitary) $\cM$-adhesive category}~\cite{Braatz:2010aa}, or whenever the structure arises as the restriction of such a category via imposing \emph{constraints} in the sense of~\cite{habel2009correctness} (see also~\cite{behrRaSiR} for further details), one may utilize rule-algebraic techniques to analyze these combinatorial structures. While we postpone a detailed analysis of precisely which prerequisites are necessary for a combinatorial species defined in terms of a species functor to also possess the structure of (a restriction of an) $\cM$-adhesive category to future work, in this paper we follow the pragmatic approach of introducing the computational theory of rule-algebra based formal moment evolution equations, and then illustrate this approach with the application example of studying the species of \emph{planar rooted binary trees (PRBTs)}.\\

In order to give an interpretation to $\cG(\lambda;\vec{\omega})$ within the theory of combinatorics, consider first the special case of setting all formal parameters $\omega_1,\dotsc,\omega_m$ to zero:
\begin{equation}\label{eq:genSeriesG}
\cG(\lambda;\vec{0}) = \sum_{n\geq 0} \tfrac{\lambda^n}{n!} \bra{}\hat{G}^n\ket{X_0}\equiv \sum_{n\geq 0}\tfrac{\lambda^n}{n!} g_n
\end{equation}
Since $\hat{G}^n\ket{X_0}$ is nothing but a distribution over all outcomes of applying $n$ rewriting steps with rules from the given rule-set in all possible ways, and with weights given by the parameters $\gamma_1,\dotsc,\gamma_n$, one may recognize $\cG(\lambda;\vec{0})$ as the exponential generating function (EGF) of some \emph{weighted combinatorial species}~\cite{bergeron1997,bergeron2013introduction}. To expose this feature more clearly, let $\{\Rightarrow_{(i)}\}_{i>0}$ denote the reachability relation on $\obj{\bfC}_{\cong}^{\times 2}$ as introduced in~\eqref{eq:defG} (with respect to the initial configuration $X_{0}\in\obj{\bfC}_{\cong}$ and the rule-set $\{R_j\}_{h=1}^n$ used to define $\hat{G}$). Then one may view $\hat{G}$ as the \emph{generator} of a (countable) set of structures $S_{\hat{G}}$,
\begin{equation}
S_{\hat{G}}:=\cup_{n>0}S_{\hat{G}}^{(n)}\,,\quad S_{\hat{G}}^{(n)}:=
\begin{cases}
\{X_0\} \quad & \text{if } n=0\\
\{ X\in \obj{\bfC}_{\cong}\mid X_0\Rightarrow_{(n)} X\}& \text{if } n>0\,.
\end{cases}
\end{equation}
Since we assume $X_0$ to be a finite object (in the sense of a finite number of $\cM$-subobjects), clearly each of the sets $S_{\hat{G}}^{(n)}$ is of finite cardinality. In addition, the coefficients $g_n=\bra{}\hat{G}^n\ket{X_0}$ are evidently of finite value as well, which in summary permits the following \emph{repartition} of the formal power series $\cG(\lambda;\vec{0})$:
\begin{equation}\label{eq:gnXDef}
\cG(\lambda;\vec{0}) = \sum_{n\geq 0} \tfrac{\lambda^n}{n!} \sum_{X\in S^{(n)}_{\hat{G}}} g_n(X)\,,\qquad g_n(X):=\bra{X}\hat{G}^n\ket{X_0}
\end{equation}
Consequently, the configurations $X\in S^{(n)}_{\hat{G}}$ may be seen as the combinatorial structures contained in the $n$-th generation, with $g_n(X)$ the \emph{weight} of a configuration $X$ in the $n$-th generation.

For generic values of $\vec{\omega}$, $\cG(\lambda;\vec{\omega})$ evaluates as follows:
\begin{equation}
\cG(\lambda;\vec{\omega})=\sum_{n\geq0}\tfrac{\lambda^n}{n!} 
\bra{}e^{\vec{\omega}\cdot\vec{\hat{O}}}\hat{G}^n\ket{X_0}
=\sum_{n\geq 0} \tfrac{\lambda^n}{n!} \sum_{X\in S^{(n)}_{\hat{G}}} g_n(X)
e^{\vec{\omega}\cdot\vec{N(X)}}\,,\quad N_i(X):=\bra{}\hat{O}_i\ket{X}\,.
\end{equation}

\subsection{Embedded discrete-time Markov chains}\label{sec:edtmcs}

Given a linear operator $\hat{G}$ for which all of the parameters $\gamma_1,\dotsc,\gamma_n\in \bR$ are positive, and with initial state $\ket{X_0}$ for some $X_0\in \obj{\bfC}_{\cong}$, one may construct from this data a CTMC according to~\eqref{eq:defCTMC}:
\begin{equation}
H=\hat{H}+H_{O}\,,\quad \hat{H}:=\rho(G)=\hat{G}\,,\quad H_O:=-\hat{\bO}(G)\,,\quad
G:=\sum_{j=1}^n \gamma_j\delta(R_j)
\end{equation}
This raises an interesting question: what is the precise interpretation of the formal power series $\cG(\lambda;\vec{\omega})$ in this particular setting? Recall first from~\eqref{eq:genSeriesG} that the coefficients $g_n:=\bra{}\hat{G}^n\ket{X_0}$ are finite real numbers (by assumption of finiteness of $X_0$ and of the rule-set defining $\hat{G}$), and due to the additional assumption $\gamma_1,\dotsc,\gamma_n$ of $\hat{G}$ made here, they are in fact \emph{positive} real numbers. We may conclude in summary that $g_0:=1$ and $g_n\in\bR_{>0}$ for all $n\geq 0$. Evidently, we also find that $g_n(X)\in\bR_{>0}$ for all $n\geq 0$ and $X\in S_{\hat{G}}^{(n)}$, which leads to the following result:
\begin{lemma}\label{lem:eDTMC}
With notations as above, let the \emph{integral transformation}\footnote{This notation is motivated by the operational theory of \emph{umbral calculus} presented in~\cite{bddlp2018}.} $\hat{\bI}$ be implicitly defined via its action $\hat{\bI}(s^n):=g_n^{-1}$ on the integration variable $s$, and let $\hat{d}$ denote the \emph{umbral transformation} of $\hat{G}$, defined via $ \hat{d}^n := \hat{\bI}((s\hat{G})^n)$  (for $n\in\bZ_{\geq0}$). Then the $\tau$-dependent distribution
\begin{equation}
\ket{\Phi(\tau)}:= e^{\tau\hat{d}}\ket{X_0}=\sum_{n\geq 0}\tfrac{\tau^n}{n!}\hat{d}^n\ket{X_0}
\end{equation}
encodes a family of distributions $\ket{\Phi_n}:=\hat{d}^n\ket{X_0}$, where $\ket{\Phi_n}$ is the result of the $n$-th step of the \emph{embedded discrete-time Markov chain (DTMC)} of the $CTMC$ $(H,\ket{X_0})$. In particular, $\ket{\Phi_n}$ is a \emph{probability distribution}, with coefficients $p_n(X):=\bra{X}\hat{d}^n\ket{X_0}$ the probability of reaching state $\ket{X}$ from the initial state $\ket{X_0}$ in $n$ steps.
\end{lemma}

The readers uncomfortable with the idea of employing some integral transformation in order to define the DTMC generator $\hat{d}$ might alternatively prefer a more direct definition of $\hat{d}$ in the following form:
\begin{equation}
\hat{d}:= \hat{G}\;\cdot\; \left(\hat{\bO}(G)\right)^{-1}_{*}\,,\quad
\left(\hat{\bO}(G)\right)^{-1}_{*}\ket{X}:=\begin{cases}
1\quad &\text{if } \hat{G}\ket{X}=0_{\bR}\\
\tfrac{1}{\bra{}(\hat{\bO}(G)\ket{X}}\ket{X} &\text{otherwise.}
\end{cases}
\end{equation}
Note that the order of $\hat{G}$ and $\left(\hat{\bO}(G)\right)^{-1}_{*}$ in the above formula is not exchangeable, and that we had to define the formal inverse $\left(\hat{\bO}(G)\right)^{-1}_{*}$ of $\hat{\bO}(G)$ such that it evaluates to $1$ (rather than the division by $0$) for $\ket{X}$ with $\hat{G}\ket{X}=0_{\bR}$. However, from a purely computational viewpoint, the formulation of Lemma~\ref{lem:eDTMC} is in practice more suitable in order to derive data about embedded DTMCs via evolution equations. 

In full analogy to the arguments for the rewriting-based CTMCs, one may introduce exponential moment-generating functions for the embedded DTMCs:
\begin{equation}
\cD(\tau;\vec{\omega}):=\bra{}e^{\vec{\omega}\cdot\vec{O}}e^{\tau\hat{d}}\ket{X_0}
\end{equation}
One of the most interesting use-cases of this construction is the situation where the set of observables satisfies the modified polynomial jump-closure condition $(\mathsf{PJC}')$ of~\eqref{eq:defPJC2} with respect to $\hat{G}$ (and thus by extension also for $\hat{d}$). Assuming once again the existence and finiteness of all moments, the variable transformations $\omega_i\to \ln x_i$ then induce the following evolution equation:
\begin{equation}\label{eq:eDTMCpgfEvo}
\tfrac{\partial}{\partial\tau}P(\tau;\vec{x})= \hat{\mathsf{d}}(\vec{x},\vec{\partial x})P(\tau;\vec{x})\,,\quad
P(\tau;\vec{x}):= \bra{}\vec{x}^{\vec{\hat{O}}}\, e^{\tau\hat{d}}\ket{X_0}\,,\quad
\hat{\mathsf{d}}(\vec{x},\vec{\partial x}):= \left.\left(
e^{ad_{\vec{\ln x}\cdot\vec{\hat{O}}}}(\hat{d})
\right)\right\vert_{\vec{\hat{O}}\to\vec{\tfrac{\partial}{\partial x}}}
\end{equation}
As we will demonstrate in the next section, for suitable choices of $\hat{G}$ and observables, the above type of evolution equation permits to \emph{statically analyze} an induced DTMC that evolves not on the original state space $\hat{\bfC}$, but instead on a state space indexed by the vectors $\vec{N(X)}$ of pattern counts (with $N_i(X):=\bra{}\hat{O}_i\ket{X}$). The interest in such types of \emph{observable-based marginalization} of the probability distribution of the embedded DTMC is that typically the evolution over the full state space $\hat{\bfC}$ would be entirely infeasible to interpret (or even to compute), reiterating that for instance in the case of the tree-based example presented in the next section, the reachable state space even after just $100$ applications of $\hat{G}$ contains already more than $10^{217}$ states.

\section{A prototypical example: planar rooted binary trees}\label{sec:PRBTs}

Trees in all their sorts and varieties are amongst some of the best-studied combinatorial structures, yet have not been considered in any detail from the viewpoint of graph rewriting theory. For the present illustration, let us consider \emph{planar rooted binary trees (PRBTs)} and disjoint unions thereof, which will be referred to as planar rooted binary forests (PRBFs). We will encode PRBTs as \emph{typed directed graphs} that satisfy certain structural constraints. Concretely, let $\mathbf{prePRBF}$ (the ``host category'' for planar rooted binary forests) be the adhesive category of directed multigraphs typed over the type-graph $T_{PRBF}$,
\begin{equation}
\mathbf{prePRBF}:=\mathbf{FinGraph}/T_{PRBF}\,,\quad T_{PRBF}:= \inputTikz{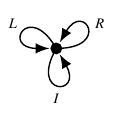}
\end{equation}
In close analogy to the fashion in which the data type of Kappa site-graphs~\cite{Boutillier:2018aa} may be encoded as recently described in~\cite{BK2020}, PRBFs may be defined as objects of $\mathbf{prePRBF}$ that satisfy the \emph{structural constraint} $\ac{c}_{PRBF}$ that is defined in terms of negative and positive constraints over the initial object $\mIO$ (i.e.\ the ``empty object'') as follows:
\gdef\sepS{0.6}
\begin{equation}\label{eq:defACpbrf}
\begin{aligned}
\ac{c}_{PRBF}&:= \ac{c}_{PRBF}^{(-)}\land \ac{c}_{PRBF}^{(+)}\,,\qquad
\ac{c}_{PRBF}^{(-)}:= \bigwedge_{N\in \cN_{PRBF}} \not \exists (\mIO\hookrightarrow N)\\ 
\cN_{PRBF}&:=\left\{\inputTikz{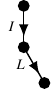}\,,\,\inputTikz{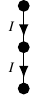}\,,\,\inputTikz{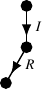}\,,\,\inputTikz{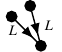}\,,\,\inputTikz{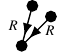}\right\}\cup \bigcup_{T,T'\in\{I,L,R\}}\left\{\inputTikz{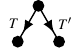}\,,\,\inputTikz{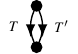}\,,\,\inputTikz{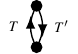}\right\}\\
\ac{c}_{PRBF}^{(+)}&:=\forall\left(\mIO\hookrightarrow 
\inputTikz{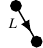}, \exists\left(\inputTikz{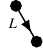}\hookrightarrow\inputTikz{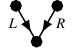}\right)\right)\bigwedge
\forall\left(\mIO\hookrightarrow 
\inputTikz{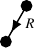}, \exists\left(\inputTikz{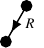}\hookrightarrow\inputTikz{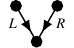}
\right)\right)\\
&\qquad\qquad\bigwedge
\bigwedge_{T\in\{L,R\}}
\forall\left(
\mIO\hookrightarrow \inputTikz{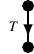}, 
\bigvee_{T'\in\{I,L,R\}}
\exists\left(\inputTikz{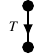}\hookrightarrow\inputTikz{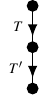}
\right)\right)
\end{aligned}
\end{equation}
Following yet again the tradition of the Kappa framework~\cite{Boutillier:2018aa} (see also~\cite{Danos2014}), let us introduce the set $\mathbf{P}_{PRBF}$ of \emph{PRBF patterns} and the set $\mathbf{S}_{PRBF}$ of \emph{states} (with the latter coinciding of course with the set of PRBFs), with the natural hierarchy $\mathbf{S}_{PRBF}\subset \mathbf{P}_{PRBF}\subset\obj{\mathbf{prePRBF}}_{\cong}$: 
\begin{equation}
\mathbf{P}_{PRBF}:=\{ X\in \obj{\mathbf{prePRBF}}_{\cong}\mid X\vDash\ac{c}_{PRBF}^{(-)}\}\,,\quad
\mathbf{S}_{PRBF}:=\{ X\in \mathbf{P}_{PRBF}\mid X\vDash\ac{c}_{PRBF}^{(+)}\}
\end{equation}
The importance of this distinction between patterns and states is that in general one may define \emph{rules} over \emph{patterns}, while \emph{states} will be the types of structures over which we will study CTMC, DTMC or combinatorial constructions. It is well-known~\cite{Danos2014} that by virtue of the properties of negative application conditions, for every $\cM$-morphism $P'\hookrightarrow P$ where $P$ is a pattern, $P'$ is a pattern as well.

We may finally define \emph{planar rooted binary trees} as elements of $\mathbf{S}_{PRBF}$ that are in addition \emph{connected graphs}. If we let $\cT_n$ denote the set of PRBTs with $(n+1)$ leaves, we thus find for example
\begin{equation}
\cT_0 :=\left\{\inputTikz{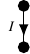}\right\}\,,\;
\cT_1 :=\left\{\inputTikz{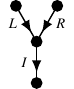}\right\}\,,\;
\cT_2 :=\left\{\inputTikz{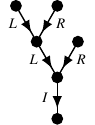}\,,\,\inputTikz{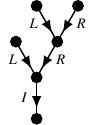}\right\}\,,\; \dotsc
\end{equation}
From hereon, we will simplify our graphical notations via omitting the vertices when drawing PRBTs, the direction of edges and also the $I$, $L$ and $R$ labels where possible, since the type of the edges may be inferred from the chosen ``standard orientation'' for the PRBT depictions:
\begin{equation}
\inputTikz{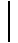}\equiv \inputTikz{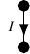}\,,\quad \inputTikz{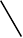}\equiv \inputTikz{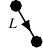}\,\quad\inputTikz{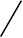}\equiv \inputTikz{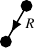}
\end{equation}

For illustration of the computational framework put forward in the present paper, we will construct and analyze PRBTs via the so-called \emph{R\'{e}my uniform generator}~\cite{remy1985procede}, starting from the root-only PRBT $|\in\cT_0$. The generator may be encoded in the present formalism as follows (where we let\footnote{Note that our choice of SqPO-type rewriting was taken purely for technical convenience, i.e.\ in order to take advantage of the slightly simpler structure of SqPO-type observables (cf.\ \eqref{eq:defSqPOjc}) and SqPO-type jump-closure (cf.\ \eqref{eq:JCgen}).} $\rho:=\rho_{\mathbf{prePRBF}}^{{SqPO}}$):
\begin{equation}\label{eq:defGR}
\begin{aligned}
\hat{G}:= \hat{G}_L+\Hat{G}_R\,,\;
\hat{G}_L&:= \inputTD{1}{\tdScale}:=\sum_{T\in\{I,L,R\}}\rho\left(\delta\left(
\inputTikz{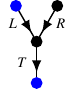}\hookleftarrow \inputTikz{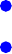}\hookrightarrow \inputTikz{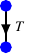}\,,\; \Shift\left(
\mIO\hookrightarrow \inputTikz{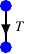}, \ac{c}_{PRBF}\right)
\right)\right)\\
\hat{G}_R&:= \inputTD{2}{\tdScale}:=\sum_{T\in\{I,L,R\}}\rho\left(\delta\left(
\inputTikz{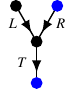}\hookleftarrow \inputTikz{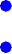}
\hookrightarrow \inputTikz{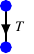}\,,\; \Shift\left(\mIO\hookrightarrow \inputTikz{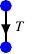}, \ac{c}_{PRBF}\right)
\right)\right)
\end{aligned}
\end{equation}
Here, in the specification of the rewriting rules, we have highlighted the vertices that are preserved by the rules (as black vertices in the compressed notation, and in {\color{blue}blue}  in the explicit notation for better readability). The application conditions for the rules are simply suitable shifts of the structural constraints $\ac{c}_{PRBT}$ to the input interfaces of the rules; this is because while one would a priori also need to consider a contribution to the application conditions that ensures satisfaction of $\ac{c}_{PRBT}$ after application of the rules (i.e.\ technically applying $\Trans$ to the $\Shift$ of $\ac{c}_{PRBT}$ from $\mIO$ to the output interfaces), one may compute that the resulting conditions are subsumed by the ones explicitly mentioned in~\eqref{eq:defGR}.

As a first consistency check, we verify that $\hat{G}$ is a \emph{uniform generator}, in the sense that
\begin{equation}
\hat{G}\ket{\,|\,} = \sum_{t\in \cT_1} 2! \ket{t}\,,\;
\forall t\in \cT_1: \hat{G}\ket{t} = \sum_{t'\in \cT_2} 3! \ket{t'}\,,\;\dotsc\,,\;
\forall t\in \cT_n: \hat{G}\ket{t}=\sum_{t'\in \cT_{n+1}}(n+2)!\ket{t'}\,.
\end{equation}
In other words, starting from an arbitrary tree $t\in\cT_n$ in ``generation'' $n$, applying $\hat{G}$ yields a uniform distribution over all trees in ``generation'' $n+1$, all with weight $(n+2)!$. 

The core computational strategy put forward in our rule-algebraic framework consists in searching for polynomially jump-closed observable sets and in computing the \emph{commutators} that occur in the various forms of evolution equations. We will focus here on the simplest form of SqPO-type pattern counting observables, namely those of the form $\hat{O}_P:=\hat{O}_{P;\ac{true}}$ (cf.\ \eqref{eq:defSqPOjc}). Since we will be exclusively interested in evaluating the action of $\hat{G}$ and of the observables on PRBTs states, our computations may be simplified to \emph{constraint-preserving} semantics in the sense of~\cite{habel2009correctness}. Under this simplified semantics, the application conditions in the rules of $\hat{G}$ are equivalent to $\ac{true}$, whence in computing SqPO-type rule compositions for the commutators, the problem simplifies drastically to the following one: a partial overlap between the input or output of a rule in $\hat{G}$ with an output or input of another rule $\hat{R}$ with application condition $\ac{true}$ is an admissible match if and only if it is an admissible match of the ``plain rules'', and if in addition the gluing ${\color{h1color}N_{21}}$ of the interfaces as in~\eqref{eq:defRcomp} satisfies the pattern constraints (i.e.\ if ${\color{h1color}N_{21}}\vDash\ac{c}_{PRBF}^{(-)}$).

Let us begin with the simplest non-trivial polynomial jump-closed set of observables for $\hat{G}$, which consists just of the observable $\hat{O}_E$ that ``counts'' edges in the trees regardless of their type:
\begin{equation}
\hat{O}_E := \inputTD{24}{\tdScale} :=\sum_{T\in\{I,L,R\}}\rho\left(\delta\left(
\inputTikz{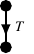}\hookleftarrow \inputTikz{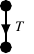}\hookrightarrow\inputTikz{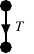}\,,\;\ac{true}
\right)\right)
\end{equation}
According to SqPO-type jump-closure and under constraint-preserving semantics (i.e.\ when acting on PRBTs), we may verify that the set $\{\hat{O}_E\}$ is indeed polynomially jump-closed with respect to $\hat{G}$:
\begin{equation}\label{eq:GoE}
(i)\quad [\hat{O}_E,\hat{G}] = 2\hat{G}\,,\qquad 
(ii)\quad \bra{}\hat{G} = 2\bra{}\hat{O}_E\,.
\end{equation}
In order to gain some intuitions for the computation technique for commutators, we present below some details on $(i)$, where $\dotsc$ denote contributions that drop out of the commutator due to sequential independence, and where we have highlighted the rules of $\hat{G}$ in {\color{h2color}orange} to show the structure of the individual rule compositions:
\begin{equation}
[\hat{O}_E,\hat{G}]= \left[
\inputTikz{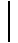}+ \inputTikz{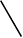} + \inputTikz{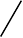}\,,\, \inputTD{3}{\tdScale}+\inputTD{4}{\tdScale}
\right]=\inputTD{5}{\tdScale}+\inputTD{6}{\tdScale}+\inputTD{7}{\tdScale}+\inputTD{8}{\tdScale}+\dotsc-\dotsc=2\hat{G}\,.
\end{equation}
This result is sufficient to perform our first moment-EGF computation:
\begin{equation}
\begin{aligned}
\cG(\lambda;\varepsilon)&:=\bra{}e^{\varepsilon\hat{O}_E}e^{\lambda\hat{G}}\ket{\,|\,}\\
\tfrac{\partial}{\partial \lambda}\cG(\lambda;\varepsilon)
&=\bra{}\left(e^{ad_{\varepsilon\hat{O}_E}}(\hat{G}\right)e^{\varepsilon\hat{O}_E}e^{\lambda\hat{G}}\ket{\,|\,}
=\sum_{q\geq0}\tfrac{1}{q!}\bra{}\left(ad_{\varepsilon \hat{O}_E}^{\circ\:q}(\hat{G})\right)e^{\varepsilon\hat{O}_E}e^{\lambda\hat{G}}\ket{\,|\,} &(via~\eqref{eq:EMGFfeG})\\
&=\big(\sum_{q\geq 0}\tfrac{(2\varepsilon)^q}{q!}\big)\bra{}\hat{G}e^{\varepsilon\hat{O}_E}e^{\lambda\hat{G}}\ket{\,|\,}
=2e^{2\varepsilon}\bra{}\hat{O}_Ee^{\varepsilon\hat{O}_E}e^{\lambda\hat{G}}\ket{\,|\,}&(via~\eqref{eq:GoE})
\end{aligned}
\end{equation}
We have thus derived an evolution equation that is solvable e.g.\ via \emph{semi-linear normal-ordering}~\cite{bdg2019}:
\begin{equation}
\left\{
\begin{aligned}
\tfrac{\partial}{\partial \lambda}\cG(\lambda;\varepsilon)&=
2e^{2\varepsilon}\tfrac{\partial}{\partial\varepsilon} \cG(\lambda;\varepsilon)\\
\cG(0;\varepsilon) &= \bra{}e^{\varepsilon\hat{O}_E}\ket{\,|\,}=e^{\varepsilon}
\end{aligned}
\right.\quad\Rightarrow \quad
\cG(\lambda;\varepsilon)=\tfrac{1}{\sqrt{e^{-2\varepsilon}-4\lambda}}
=\sum_{n\geq 0}\tfrac{\lambda^n}{n!}\left(\tfrac{(2n)!}{n!}e^{\varepsilon(2n+1)}\right)\,.
\end{equation}
Unsurprisingly, the final result for $\cG(\lambda;\varepsilon)$ expresses that all PRBTs in ``generation'' $n$ have the same overall number of edges (i.e.\ $2n+1$), and we invite the readers to verify via explicitly computing $g_n:=\bra{}\hat{G}^n\ket{\,|\,}$ for small values of $n$ that $g_n=(2n)!/n!$, which is obtained alternatively via specializing $\cG(\lambda;\varepsilon)$ to $\varepsilon=0$  (with $\cG(\lambda):=\cG(\lambda;0)=\sum_{n\geq 0}\tfrac{\lambda}{n!}g_n$).

The true test of utility of the rule-algebraic methods is of course whether or not it is possible to compute evolution equations for more intricate observables, since in the case of $\hat{O}_E$ it would have been possible to derive the evolution equations and the moment EGF via heuristics. To this end, it will prove useful to introduce some auxiliary results to facilitate dealing with nested commutator equations.

\begin{lemma}\label{lem:adj}
For arbitrary $\hat{R}=\rho(\delta(R))$ and $\{\hat{O}_1,\dotsc,\hat{O}_m\}$ a set of observables, nested actions of the observables on $\hat{R}$ are \emph{multi-linear} and \emph{symmetric} in the following sense (for any permutation $\sigma\in S_m$):
\begin{equation}
ad_{\omega_1\hat{O}_1}\left(ad_{\omega_2\hat{O}_2}\left(\dotsc\left(ad_{\omega_m\hat{O}_m}(\hat{R})\right)\dotsc\right)\right) =
(\omega_1\cdots \omega_m)
ad_{\hat{O}_{\sigma(1)}}\left(ad_{\hat{O}_{\sigma(1)}}\left(\dotsc\left(ad_{\hat{O}_{\sigma(m)}}(\hat{R})\right)\dotsc\right)\right)\,.
\end{equation}
Thus in particular 
\begin{equation}
(i)\quad e^{ad_{\omega_1\hat{O}_1}}[\hat{O}_2,\hat{R}] = [\hat{O}_2,e^{ad_{\omega_1\hat{O}_1}}(\hat{R})]\,,
\quad (ii)\quad 
e^{ad_{\omega_1\hat{O}_1}}\left(e^{ad_{\omega_1\hat{O}_2}}(\hat{R})\right)
=e^{ad_{\omega_1\hat{O}_2}}\left(e^{ad_{\omega_1\hat{O}_1}}(\hat{R})\right)
\,.
\end{equation}
\end{lemma}

As a case study, we will now present a body of results on a set of observables that is non-trivially polynomially jump-closed with respect to $\hat{G}$. We will represent by convention an observable $\hat{O}_P$ simply by the pattern $P$ (motivated by the fact that the rule underlying $\hat{O}_P$ is an identity rule):
\begin{equation}
\hat{O}_{P1}:=
\inputTD{9}{\tdScale}
\equiv\sum_{T\in \{I,L,R\}}\!\!\inputTD{25}{\tdScale}\,,\;
\hat{O}_{P2}:= \inputTD{26}{\tdScale}
\equiv\sum_{T\in \{I,L,R\}}\!\!\inputTD{27}{\tdScale}\,,\;
\hat{O}_{P3}:= \inputTD{28}{\tdScale}
\equiv\sum_{T\in \{I,L,R\}}\!\!\inputTD{29}{\tdScale}
\end{equation}
We will need a considerable number of commutator equations, which could in principle ultimately be performed automatically via our tool\footnote{\url{https://gitlab.com/nicolasbehr/ReSMT} (GitHub), \url{https://nicolasbehr.gitlab.io/ReSMT} (documentation)} \texttt{ReSMT}~\cite{behr2020commutators}, but which were computed manually here. For the sake of illustration, we present the computation of $[\hat{O}_{P1},\hat{G}]$ in some detail below (where $\dotsc$ denote contributions that cancel from the commutator due to sequential independence):
\begin{equation}
\begin{aligned}
[\hat{O}_{P1},\hat{G}]&=\left[\inputTD{25}{\tdScale},
\inputTD{3}{\tdScale}+\inputTD{4}{\tdScale}
\right]\\
&= 
\inputTD{11}{\tdScale} +
\inputTD{10}{\tdScale}
+\inputTD{12}{\tdScale}
+\inputTD{13}{\tdScale}
+\inputTD{18}{\tdScale}
+\inputTD{17}{\tdScale}
+\inputTD{19}{\tdScale}
+\inputTD{20}{\tdScale}+\dotsc\\
&\quad- 
\inputTD{15}{\tdScale}-\inputTD{15}{\tdScale}-\inputTD{16}{\tdScale}
-\inputTD{21}{\tdScale}-\inputTD{22}{\tdScale}-\inputTD{23}{\tdScale}- \dotsc =\hat{G}
\end{aligned}
\end{equation}
While $\hat{O}_{P1}$ has thus a similarly simple closure property under commutators with $\hat{G}$ as was the case with $\hat{O}_E$, the observables $\hat{O}_{P2}$ and $\hat{O}_{P3}$ in contrast present an interesting form of higher-order commutator-closure. Due to the complexity of the computations, we present here only the final results:
\begin{equation}
\begin{aligned}\gdef\tdScale{0.4}
[\hat{O}_{P2},\hat{G}]&=\inputTD{30}{\tdScale}+\inputTD{31}{\tdScale}+\inputTD{32}{\tdScale}-\inputTD{33}{\tdScale} -\inputTD{34}{\tdScale}
\qquad\qquad\qquad   {\color{h2color}\hat{R}_{P3'}}:=\inputTD{42}{\tdScale}\\
[\hat{O}_{P3},\hat{G}]&= \inputTD{35}{\tdScale}+\inputTD{36}{\tdScale}
+\inputTD{37}{\tdScale}+\inputTD{38}{\tdScale}
-\inputTD{39}{\tdScale}-\inputTD{40}{\tdScale}
-\inputTD{41}{\tdScale}-{\color{h2color}\hat{R}_{P3'}}\\
[\hat{O}_{P2},[\hat{O}_{P2},\hat{G}]]
&=[\hat{O}_{P2},\hat{G}]\,,\quad [\hat{O}_{P2},[\hat{O}_{P3},\hat{G}]]
=[\hat{O}_{P3},\hat{G}] + \hat{R}_{P3}\\
[\hat{O}_{P3},[\hat{O}_{P3},\hat{G}]]
&=[\hat{O}_{P3},\hat{G}] + 2\hat{R}_{P3'}\,,\quad
[\hat{O}_{P2},\hat{R}_{P3'}]= 0\,,\quad
[\hat{O}_{P3},\hat{R}_{P3'}]= -\hat{R}_{P3'}\\
\bra{}[\hat{O}_{P2},\hat{G}]&=\bra{}(3\hat{O}_{P1}-2\hat{O}_{P2})\,,\quad
\bra{}[\hat{O}_{P3},\hat{G}] = \bra{}(4\hat{O}_{P2}-3\hat{O}_{P3})\,,\quad
\bra{}\hat{R}_{P3'}=\bra{}\hat{O}_{P3}
\end{aligned}
\end{equation}
Summarizing all commutator results thus far, we may conclude that the observables $\{\hat{O}_E,\hat{O}_{P1},\hat{O}_{P2},\hat{O}_{P3}\}$ are polynomially jump-closed with respect to $\hat{G}$, with the closure involving up to triple commutators. We may then formulate the following moment-EGF evolution equation (with $\vec{\omega}:=(\varepsilon,\gamma,\mu,\nu)$):
\begin{subequations}
\begin{align}
\cG(\lambda;\vec{\omega})
&:=\bra{}
e^{\vec{\omega}\cdot\vec{\hat{O}}}
e^{\lambda\hat{G}}\ket{\,|\,}\,,
\quad \vec{\omega}\cdot \vec{\hat{O}}
:= \varepsilon\hat{O}_E+\gamma\hat{O}_{P1}+\mu\hat{O}_{P2}+\nu\hat{O}_{P3}\\
\tfrac{\partial}{\partial \lambda}\cG(\lambda;\vec{\omega})
&=\bra{}\left(e^{ad_{\vec{\omega}\cdot\vec{\hat{O}}}}(\hat{G})\right)
e^{\vec{\omega}\cdot\vec{\hat{O}}}e^{\lambda\hat{G}}\ket{\,|\,}
\overset{(*)}{=}
\bra{}
\left(e^{ad_{\nu\hat{O}_{P3}}}
\left(e^{ad_{\mu\hat{O}_{P2}}}
\left({\color{h1color}e^{ad_{\varepsilon\hat{O}_{E}+\gamma\hat{O}_{P1}}}(\hat{G})}\right)
\right)\right)
e^{\vec{\omega}\cdot\vec{\hat{O}}}e^{\lambda\hat{G}}\ket{\,|\,}\\
&={\color{h1color}e^{2\varepsilon+\gamma}}\bra{}
\left(e^{ad_{\nu\hat{O}_{P3}}}
\left({\color{h2color}e^{ad_{\mu\hat{O}_{P2}}}(\hat{G})}\right)
\right)
e^{\vec{\omega}\cdot\vec{\hat{O}}}e^{\lambda\hat{G}}\ket{\,|\,}\\
&={\color{h1color}e^{2\varepsilon+\gamma}}\bra{}
\left({\color{h4color}e^{ad_{\nu\hat{O}_{P3}}}
\left({\color{h2color}\hat{G}+(e^{\mu}-1)[\hat{O}_{P2},\hat{G}]}\right)}
\right)
e^{\vec{\omega}\cdot\vec{\hat{O}}}e^{\lambda\hat{G}}\ket{\,|\,}\\
&\begin{array}{l}
\!\!=
{\color{h1color}e^{2\varepsilon+\gamma}}\bra{}
\big({\color{h4color}
\hat{G}+(e^{\mu}-1)[\hat{O}_{P2},\hat{G}]}\\
\qquad\qquad{\color{h4color}+e^{\mu}(e^{\nu}-1)[\hat{O}_{P3},\hat{G}]+(e^{\nu}-1)(e^{\mu}-e^{-\nu})\hat{R}_{P3'}}
\big)
e^{\vec{\omega}\cdot\vec{\hat{O}}}e^{\lambda\hat{G}}\ket{\,|\,}
\end{array}\\
&\begin{array}{l}
\!\!=
e^{2\varepsilon+\gamma}\bra{}\big(
2\hat{O}_E+3(e^{\mu}-1)\hat{O}_{P1}+(4 e^{\mu+\nu}-6e^{\mu}+2)\hat{O}_{P2}\\
\qquad\qquad
+(3e^{\mu}+e^{-\nu}-3e^{\mu+\nu}-1)\hat{O}_{P3}
\big)e^{\vec{\omega}\cdot\vec{\hat{O}}}e^{\lambda\hat{G}}\ket{\,|\,}
\end{array}\\
&\begin{array}{l}
\!\!=
e^{2\varepsilon+\gamma}\bra{}\big(
2\tfrac{\partial}{\partial \varepsilon}
+3(e^{\mu}-1)\tfrac{\partial}{\partial \gamma}
+(4 e^{\mu+\nu}-6e^{\mu}+2)\tfrac{\partial}{\partial \mu}\\
\qquad\qquad
+(3e^{\mu}+e^{-\nu}-3e^{\mu+\nu}-1)\tfrac{\partial}{\partial \nu}
\big)e^{\vec{\omega}\cdot\vec{\hat{O}}}e^{\lambda\hat{G}}\ket{\,|\,}
\end{array}
\end{align}
\end{subequations}
Here, in the step marked $(*)$, we took advantage of the commutativity of the adjoint action of observables according to Lemma~\ref{lem:adj}, and each of the subsequent lines amounts to evaluating the highlighted adjoint actions utilizing the formulas for the commutators, with the last step resulting from applying SqPO-type jump-closure in the sense of~\eqref{lem:adj}.

Granted that the derivation of the evolution equation for $\cG(\lambda;\vec{\omega})$ is somewhat involved, one may extract from it a very interesting insight via a transformation of variables $\omega_i\to \ln x_i$ (which entails that $\frac{\partial}{\partial \omega_i}\to x_i \frac{\partial}{\partial x_i}$), and collecting coefficients for the operators $\hat{n}_i:= x_i \frac{\partial}{\partial x_i}$:
\begin{equation}
\begin{aligned}
&\tfrac{\partial}{\partial \lambda}\cG(\lambda;\vec{\ln x})
=\hat{\mathsf{D}}\cG(\lambda;\vec{\ln x})\\
&\hat{\mathsf{D}}=
x_{\varepsilon}^2x_{\nu}
(2\hat{n}_{\varepsilon}-3\hat{n}_{\gamma}+2\hat{n}_{\mu}-\hat{n}_{\nu})
+x_{\varepsilon}^2x_{\nu}x_{\mu}
(3\hat{n}_{\gamma}-6\hat{n}_{\mu}+3\hat{n}_{\nu})
+x_{\varepsilon}^2x_{\nu}x_{\mu}^2
(4\hat{n}_{\mu}-3\hat{n}_{\nu})
+x_{\varepsilon}^2\hat{n}_{\nu}
\end{aligned}
\end{equation}
Upon closer inspection, we find that the coefficients involving the operators $\hat{n}_i$ sum up to $2\hat{n}_E$, which is why $\hat{\mathsf{d}}:=\hat{\mathsf{D}}(2\hat{n}_E)_{*}^{-1}$ qualifies as the generator of a DTMC in the sense of~\eqref{eq:eDTMCpgfEvo}. We have thus succeeded in statically extracting information from the combinatorial species of PRBTs on the joint pattern count distribution of patterns $P2$ and $P3$, with some illustrative examples plotted in Figure~\ref{fig:trees}.

\begin{figure}
    \centering
    {$\begin{array}{c}\gdef\pScale{0.18}
      \includegraphics[scale=\pScale]{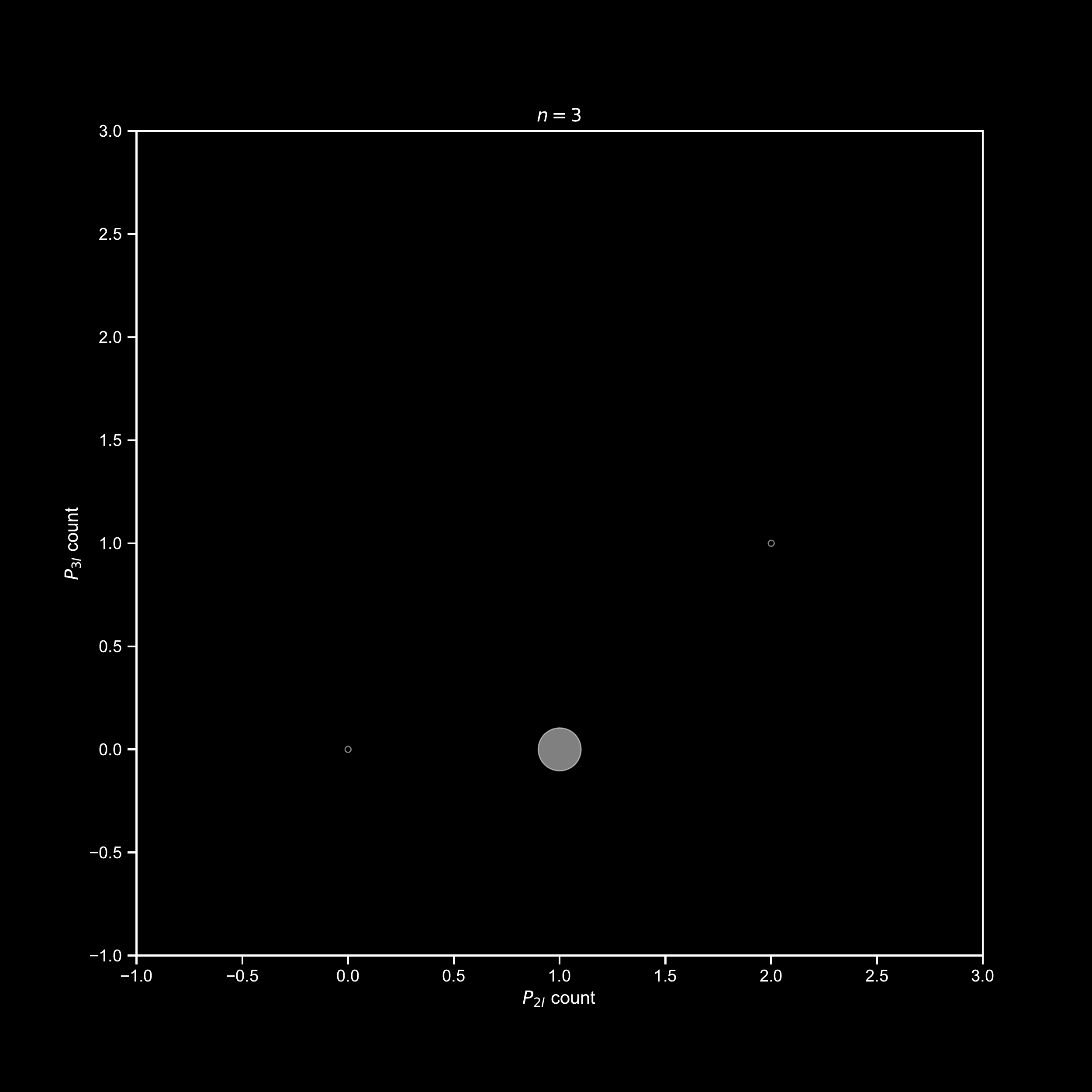}$\!$
      \includegraphics[scale=\pScale]{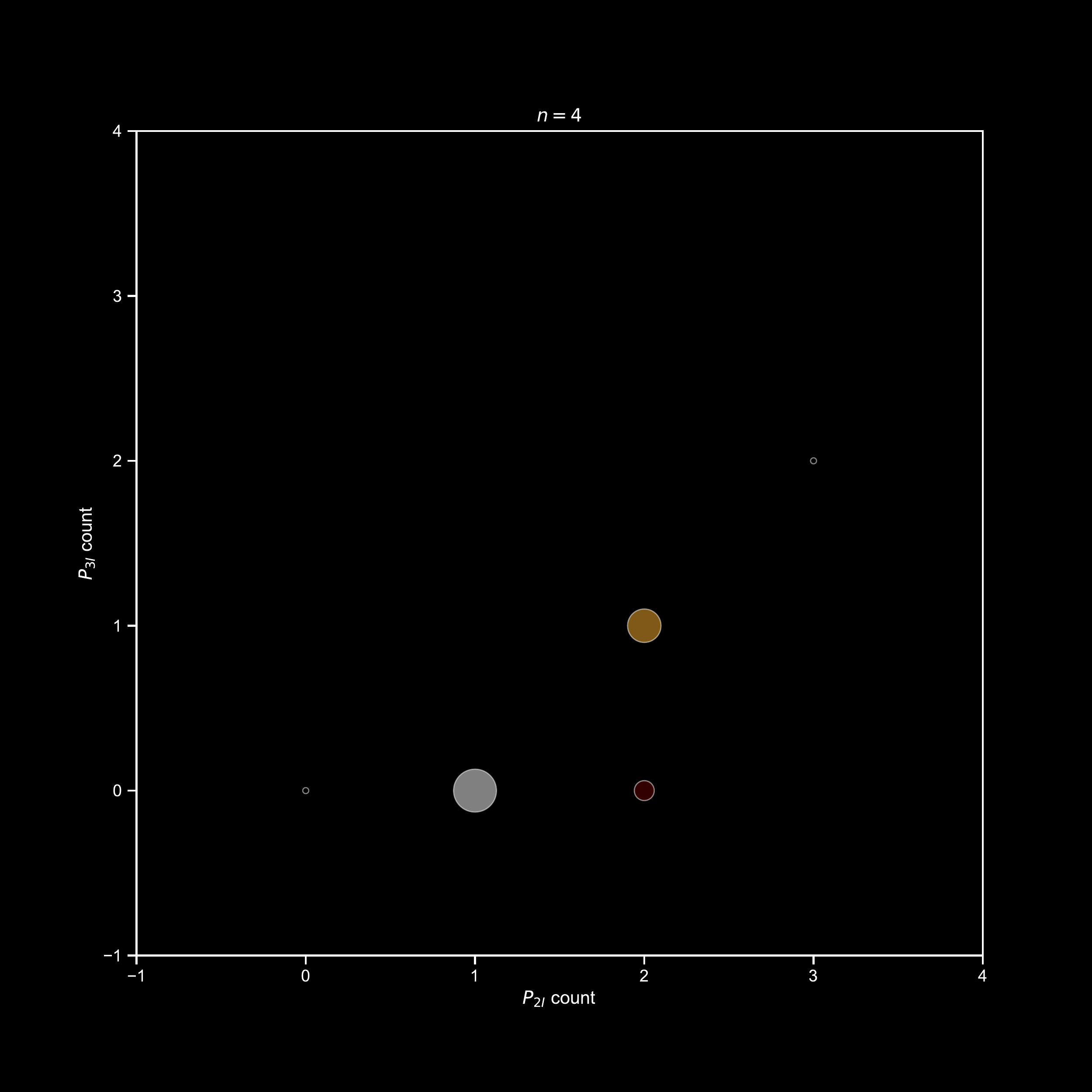}$\!$
      \includegraphics[scale=\pScale]{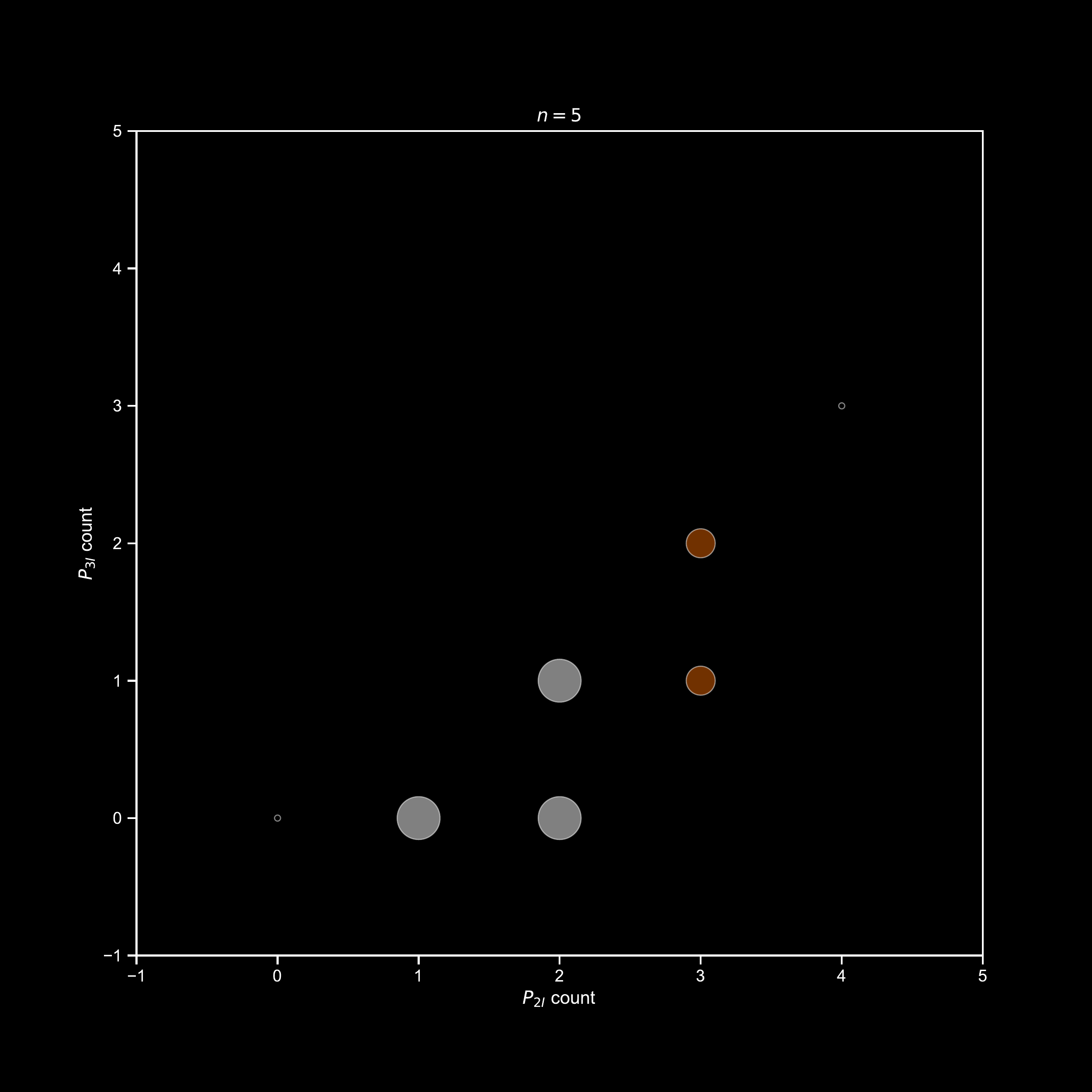}
      \\[-0.5em]
      \includegraphics[scale=\pScale]{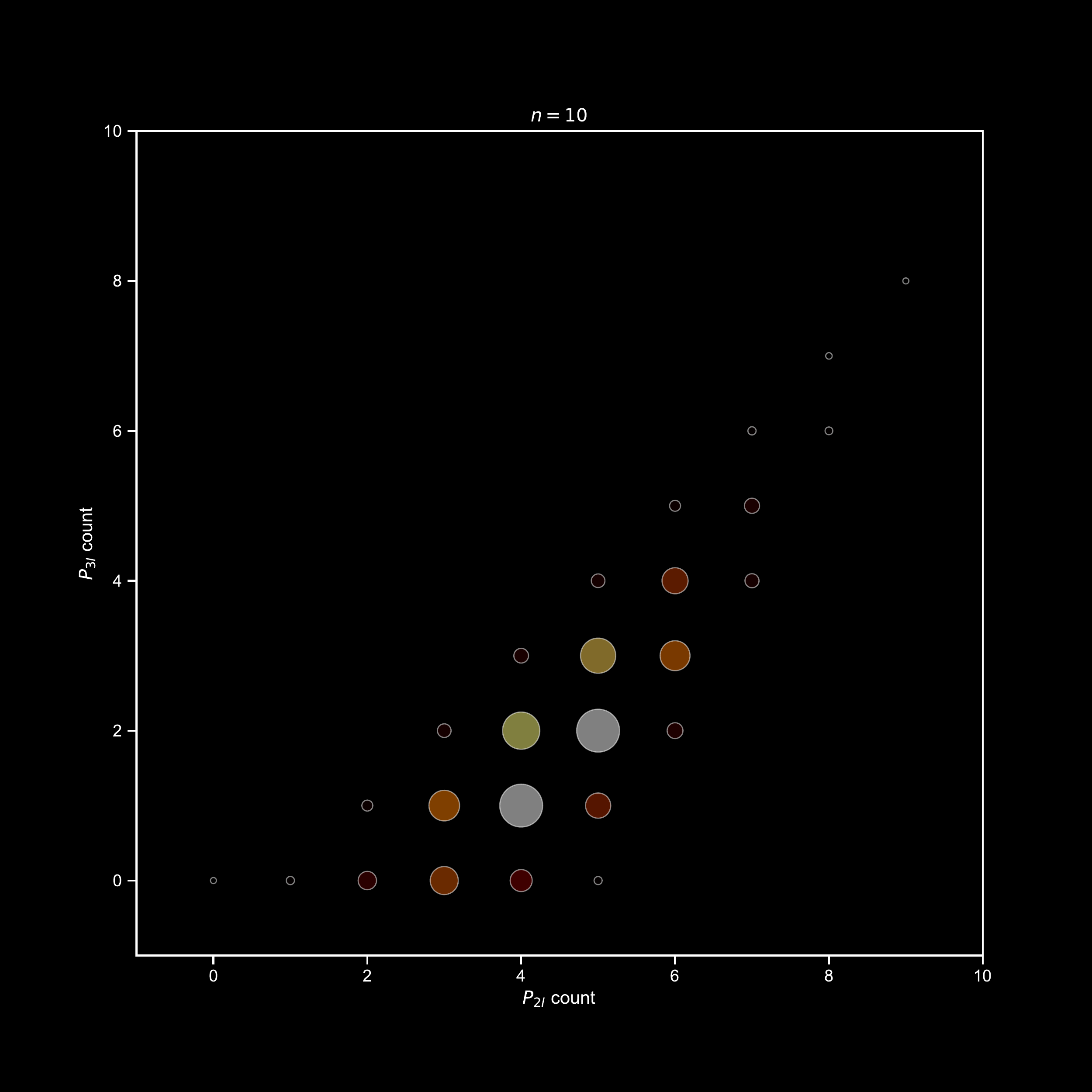}$\!$
      \includegraphics[scale=\pScale]{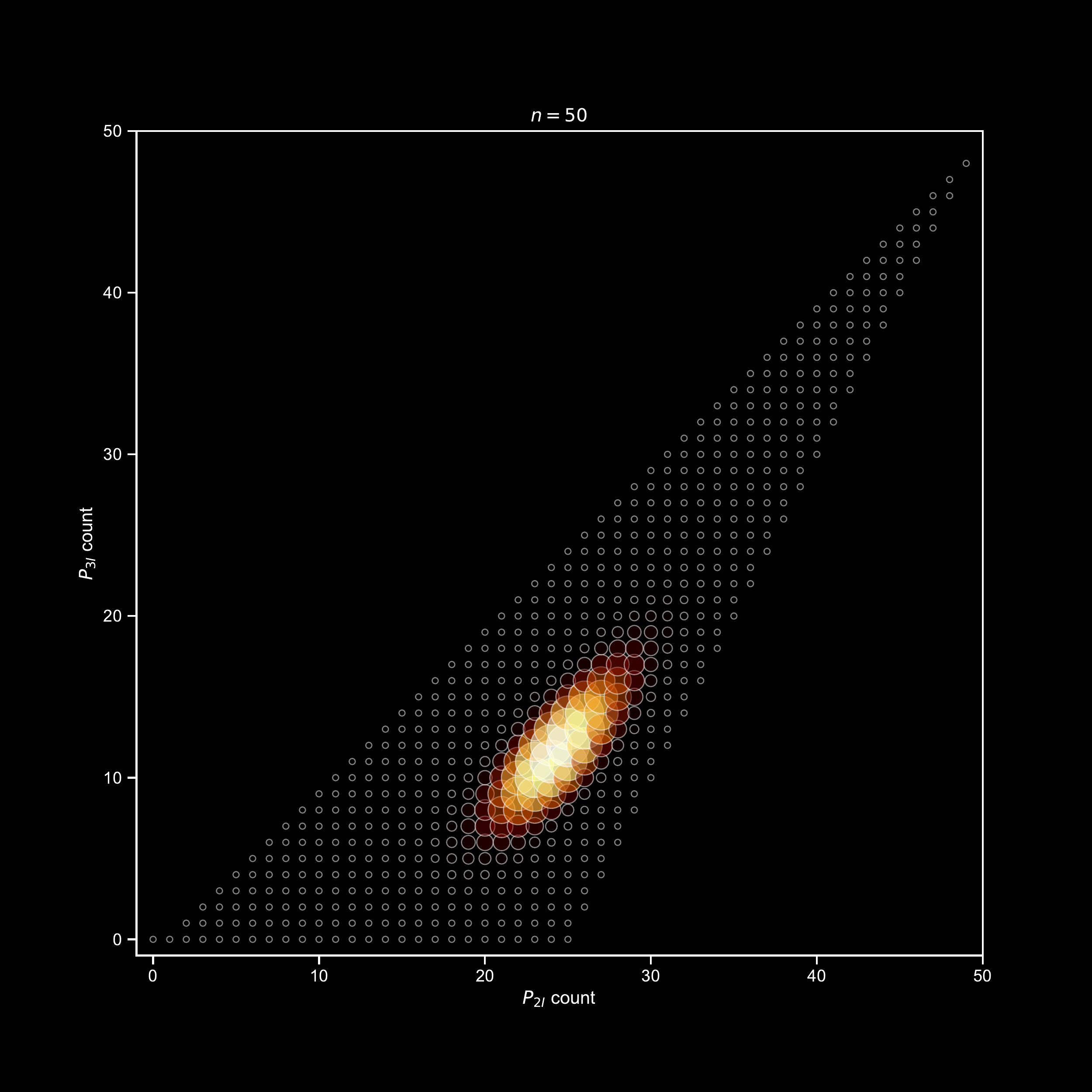}$\!$
      \includegraphics[scale=\pScale]{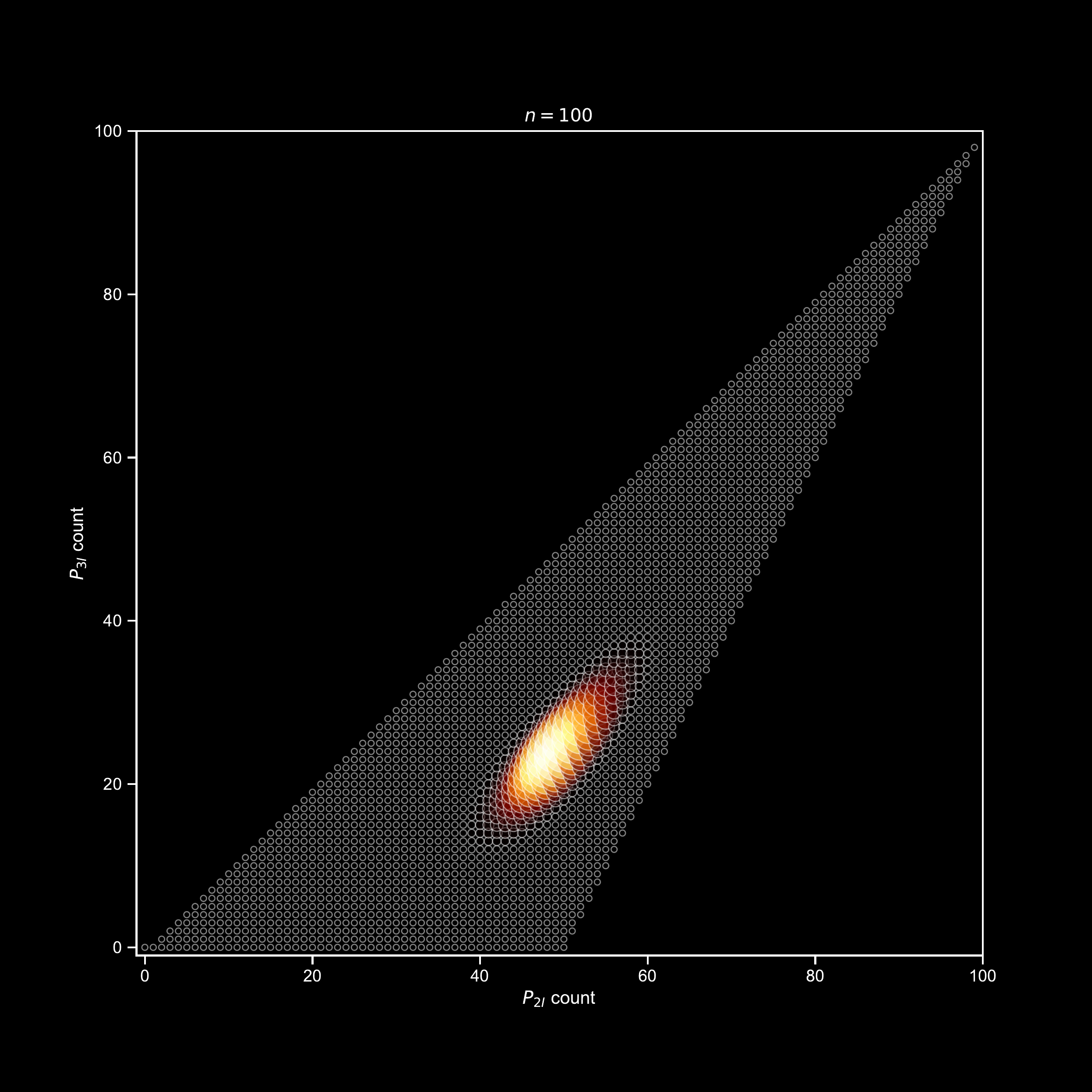}
    \end{array}$}
    \caption{\label{fig:trees}Probability distribution of \textbf{pattern counts} of pattern ${P2}$ (horizontal axis) and of pattern ${P3}$ (vertical axis) in $\hat{\mathsf{d}}^n\ket{\,|\,}$ for $n=3,4,5,10,50,100$, and with both point area and color indicating the value of the probability (larger area and brighter color for higher probability values); \href{http://nicolasbehr.com/files/videos/anim-binarytreepatterns.gif}{animated version}).}
\end{figure} 

\section{Conclusion and outlook}

Modern rewriting theory in the framework of $\cM$-adhesive categories and with conditions on objects and rewriting rules~\cite{habel2009correctness,ehrig2014mathcal,Braatz:2010aa} in either DPO-~\cite{ehrig:2006aa} or SqPO-semantics~\cite{Corradini2006} provides a powerful framework capable of encoding many types of rewriting systems over graph-like structures of practical interest. However, in the theory of Markov chains~\cite{norris} and in particular in the theory of enumerative combinatorics~\cite{FlajoletSedgewick}, even though graph-like structures and their manipulations play a prominent role, it is rarely if ever the case that such situations are analyzed via rewriting-theoretical methods. While the traditional focus of rewriting theory had been predominantly upon analyzing traces of rewriting systems and their causal properties, the aforementioned fields instead require a form of reasoning that may be intuitively described as requiring the analysis of \emph{all} possible rewriting traces of a given system, albeit typically under a certain form of \emph{projection} induced via only tracking the behavior of \emph{observables} along the traces. Over the course of a long-term research project aimed at identifying the precise conceptual and mathematical prerequisites to implement such a type of reasoning, we have developed the \emph{rule algebra} framework~\cite{bdg2016,bp2018,bp2019-ext,nbSqPO2019}, which we recently were able to extend to the setting of compositional rewriting theories over $\cM$-adhesive categories for rules with conditions in both DPO- and SqPO-semantics in~\cite{behrRaSiR,BK2020}. The rule algebra approach refocuses the analysis of rewriting systems from derivation traces to rule compositions, with the rule algebra product of two algebra elements encoding all possibilities to compose the underlying rules. Together with a suitable notion of representation theory, this approach permits a novel formalization of continuous-time Markov chains (CTMCs) that arise from stochastic rewriting systems~\cite{bdg2016,bp2019-ext,nbSqPO2019,bdg2019,BK2020}.

The main contribution of the present paper is an extension of the rule-algebraic framework for the construction and analysis of rewriting-based CTMCs to the setting of weighted combinatorial structures (Section~\ref{sec:wcs}). We moreover exhibit the notion of \emph{embedded discrete-time Markov chains (eDTMCs)} that are associated to each rewriting-based CTMC (Section~\ref{sec:edtmcs}). Both constructions are strongly motivated from the methodology advocated in our earlier work~\cite{bdg2019,BK2020} of focussing the analysis of CTMCs upon exponential moment-generating functions (EMGFs) for ensembles of pattern-counting observables. We demonstrate here that analogous formulations yield valuable insights also in the combinatorics and eDTMC settings, with a worked example in the setting of the rewriting of planar rooted binary trees (PRBTs) provided in Section~\ref{sec:PRBTs} illustrating the potential of these novel methods. While we certainly do not claim that our rewriting- and rule-algebra-based methods would necessarily result in entirely unexpected results on PRBTs (which in this particular case are in fact likely to be derivable along the lines of the work of Rowland~\cite{Rowland2010} that is based upon the combinatorial inclusion-exclusion principle), the main motivations for the rule-algebraic approach are its \emph{universality}, its \emph{agnosticism} and its potential for \emph{automation} of the requisite computations: we may (at least formally) specify the various notions of generating functions for any rewriting system that satisfies the requirements of the rule algebra framework, we do not need to employ any sophisticated combinatorics arguments to perform the computations, and ultimately all rule-algebraic computations may be performed automatically via implementations such as our (as of yet experimental) \texttt{ReSMT} framework~\cite{behr2020commutators}. Conversely, the ultimate motivation of the present work consisted in opening the rich algorithmic toolkit of enumerative and analytic combinatorics~\cite{FlajoletSedgewick} for application to analyzing stochastic rewriting systems, with the goal of developing novel static analysis techniques in particular in the settings of biochemical~\cite{Boutillier:2018aa} and organo-chemical~\cite{Andersen2016} graph rewriting. Crucially, the standard combinatorics intuition as described in the seminal book~\cite[Ch.~III.1.2]{FlajoletSedgewick} \emph{``[\ldots] the eventual goal of multivariate enumeration is the quantification of properties present with high regularity in large random structures.''} has to be reconsidered in the aforementioned chemical rewriting settings, since repeated applications of rewriting rules will \emph{not} necessarily always result in increasingly large graphical structures. However, no part of the rule-algebraic specification of evolution-equations was even remotely based on any asymptotic arguments, which is why we believe this novel viewpoint will ultimately constitute a valuable addition to the computational toolkit in the applied sciences and potentially even in combinatorics itself.


\begin{thebibliography}{10}
\providecommand{\bibitemdeclare}[2]{}
\providecommand{\surnamestart}{}
\providecommand{\surnameend}{}
\providecommand{\urlprefix}{Available at }
\providecommand{\url}[1]{\texttt{#1}}
\providecommand{\href}[2]{\texttt{#2}}
\providecommand{\urlalt}[2]{\href{#1}{#2}}
\providecommand{\doi}[1]{doi:\urlalt{http://dx.doi.org/#1}{#1}}
\providecommand{\bibinfo}[2]{#2}

\bibitemdeclare{inproceedings}{Andersen2016}
\bibitem{Andersen2016}
\bibinfo{author}{J.L. \surnamestart Andersen\surnameend},
  \bibinfo{author}{C. \surnamestart Flamm\surnameend},
  \bibinfo{author}{D. \surnamestart Merkle\surnameend} \&
  \bibinfo{author}{P.F. \surnamestart Stadler\surnameend}
  (\bibinfo{year}{2016}): \emph{\bibinfo{title}{{A Software Package for
  Chemically Inspired Graph Transformation}}}.
\newblock In: %
  {\sl \bibinfo{booktitle}{Graph Transformation (ICGT 2016)}}, {\sl
  \bibinfo{series}{LNCS,}} \bibinfo{volume}{9761}, pp. \bibinfo{pages}{73--88},
  \doi{10.1007/978-3-319-40530-8_5}.

\bibitemdeclare{inproceedings}{nbSqPO2019}
\bibitem{nbSqPO2019}
\bibinfo{author}{N. \surnamestart Behr\surnameend} (\bibinfo{year}{2019}):
  \emph{\bibinfo{title}{{Sesqui-Pushout Rewriting: Concurrency, Associativity
  and Rule Algebra Framework}}}.
\newblock In: %
 {\sl
  \bibinfo{booktitle}{{Proceedings of GCM 2019}}}, {\sl
  \bibinfo{series}{EPTCS}}
  \bibinfo{volume}{309}, pp.
  \bibinfo{pages}{23--52}, \doi{10.4204/eptcs.309.2}.

\bibitemdeclare{inproceedings}{bdg2016}
\bibitem{bdg2016}
\bibinfo{author}{N. \surnamestart Behr\surnameend},
  \bibinfo{author}{V. \surnamestart Danos\surnameend} \&
  \bibinfo{author}{I. \surnamestart Garnier\surnameend}
  (\bibinfo{year}{2016}): \emph{\bibinfo{title}{Stochastic mechanics of graph
  rewriting}}.
\newblock In: {\sl \bibinfo{booktitle}{Proceedings of {LICS} {'}16}},
  \bibinfo{publisher}{{ACM} Press}, \doi{10.1145/2933575.2934537}.

\bibitemdeclare{article}{bdg2019}
\bibitem{bdg2019}
\bibinfo{author}{N. \surnamestart Behr\surnameend},
  \bibinfo{author}{V. \surnamestart Danos\surnameend} \&
  \bibinfo{author}{I. \surnamestart Garnier\surnameend}
  (\bibinfo{year}{2020}): \emph{\bibinfo{title}{{Combinatorial Conversion and
  Moment Bisimulation for Stochastic Rewriting Systems}}}.
\newblock {\sl \bibinfo{journal}{{LMCS}}}
  \bibinfo{volume}{{16, Issue 3}}.
\newblock \urlprefix\url{https://lmcs.episciences.org/6628}.

\bibitemdeclare{article}{bddlp2018}
\bibitem{bddlp2018}
\bibinfo{author}{N. \surnamestart Behr\surnameend},
  \bibinfo{author}{G. \surnamestart Dattoli\surnameend},
  \bibinfo{author}{G{\a'e}rard~H.E. \surnamestart Duchamp\surnameend},
  \bibinfo{author}{S. \surnamestart Licciardi\surnameend} \&
  \bibinfo{author}{K.A. \surnamestart Penson\surnameend}
  (\bibinfo{year}{2019}): \emph{\bibinfo{title}{{Operational Methods in the
  Study of Sobolev-Jacobi Polynomials}}}.
\newblock {\sl \bibinfo{journal}{{Mathematics}}}
  \bibinfo{volume}{7}(\bibinfo{number}{2}), p. \bibinfo{pages}{124},
  \doi{10.3390/math7020124}.

\bibitemdeclare{article}{behrRaSiR}
\bibitem{behrRaSiR}
\bibinfo{author}{N. \surnamestart Behr\surnameend} \&
  \bibinfo{author}{J. \surnamestart Krivine\surnameend}
  (\bibinfo{year}{2019}): \emph{\bibinfo{title}{{Compositionality of Rewriting
  Rules with Conditions}}}.
\newblock {\sl
  \bibinfo{journal}{\href{https://arxiv.org/abs/1904.09322}{arXiv:1904.09322}}}.

\bibitemdeclare{inproceedings}{BK2020}
\bibitem{BK2020}
\bibinfo{author}{N. \surnamestart Behr\surnameend} \&
  \bibinfo{author}{J. \surnamestart Krivine\surnameend}
  (\bibinfo{year}{2020}): \emph{\bibinfo{title}{{Rewriting theory for the life
  sciences: A unifying framework for CTMC semantics}}}.
\newblock In: {\sl
  \bibinfo{booktitle}{Graph Transformation (ICGT 2020)}}, {\sl \bibinfo{series}{LNCS}} \bibinfo{volume}{12150},  \doi{10.1007/978-3-030-51372-6}.

\bibitemdeclare{inproceedings}{behr2020commutators}
\bibitem{behr2020commutators}
\bibinfo{author}{N. \surnamestart Behr\surnameend}, \bibinfo{author}{R.
  \surnamestart Heckel\surnameend} \& \bibinfo{author}{M. \surnamestart
  Ghaffari~Saadat\surnameend} (\bibinfo{year}{2020}):
  \emph{\bibinfo{title}{Efficient Computation of Graph Overlaps for Rule
  Composition: Theory and Z3 Prototyping}}.
\newblock In %
:  {\sl
  \bibinfo{booktitle}{Proceedings of GCM 2020}}, {\sl
  \bibinfo{series}{EPTCS}}
  \bibinfo{volume}{330}, pp.
  \bibinfo{pages}{126--144}, \doi{10.4204/EPTCS.330.8}.

\bibitemdeclare{inproceedings}{bp2018}
\bibitem{bp2018}
\bibinfo{author}{N. \surnamestart Behr\surnameend} \&
  \bibinfo{author}{P. \surnamestart Sobocinski\surnameend}
  (\bibinfo{year}{2018}): \emph{\bibinfo{title}{{Rule Algebras for Adhesive
  Categories}}}.
\newblock In: {\sl
  \bibinfo{booktitle}{Proceedings of CSL 2018}}, {\sl \bibinfo{series}{LIPIcs}} \bibinfo{volume}{119}, pp. \bibinfo{pages}{11:1--11:21}, \doi{10.4230/LIPIcs.CSL.2018.11}.

\bibitemdeclare{article}{bp2019-ext}
\bibitem{bp2019-ext}
\bibinfo{author}{N. \surnamestart Behr\surnameend} \&
  \bibinfo{author}{P. \surnamestart Sobocinski\surnameend}
  (\bibinfo{year}{2020}): \emph{\bibinfo{title}{{Rule Algebras for Adhesive
  Categories (extended journal version)}}}.
\newblock {\sl \bibinfo{journal}{{LMCS}}}
  \bibinfo{volume}{{Volume 16, Issue 3}}.
\newblock \urlprefix\url{https://lmcs.episciences.org/6615}.

\bibitemdeclare{book}{bergeron1997}
\bibitem{bergeron1997}
\bibinfo{author}{F. \surnamestart Bergeron\surnameend},
  \bibinfo{author}{G. \surnamestart Labelle\surnameend} \&
  \bibinfo{author}{P. \surnamestart Leroux\surnameend}
  (\bibinfo{year}{1997}): \emph{\bibinfo{title}{{Combinatorial Species and
  Tree-like Structures}}}.
\newblock \bibinfo{series}{Encyclopedia of Mathematics and its Applications},
  \doi{10.1017/CBO9781107325913}.

\bibitemdeclare{article}{bergeron2013introduction}
\bibitem{bergeron2013introduction}
\bibinfo{author}{F. \surnamestart Bergeron\surnameend},
  \bibinfo{author}{G. \surnamestart Labelle\surnameend} \&
  \bibinfo{author}{P. \surnamestart Leroux\surnameend}
  (\bibinfo{year}{2013}): \emph{\bibinfo{title}{Introduction to the Theory of
  Species of Structures}}.
    \newblock \urlprefix\url{http://bergeron.math.uqam.ca/wp-content/uploads/2013/11/book.pdf}.


\bibitemdeclare{article}{Boutillier:2018aa}
\bibitem{Boutillier:2018aa}
\bibinfo{author}{P. \surnamestart Boutillier\surnameend},
  \bibinfo{author}{M. \surnamestart Maasha\surnameend},
  \bibinfo{author}{X. \surnamestart Li\surnameend},
  \bibinfo{author}{H.F. \surnamestart Medina-Abarca\surnameend},
  \bibinfo{author}{J. \surnamestart Krivine\surnameend},
  \bibinfo{author}{J. \surnamestart Feret\surnameend},
  \bibinfo{author}{I. \surnamestart Cristescu\surnameend},
  \bibinfo{author}{A.G .\surnamestart Forbes\surnameend} \&
  \bibinfo{author}{W. \surnamestart Fontana\surnameend}
  (\bibinfo{year}{2018}): \emph{\bibinfo{title}{{The Kappa platform for
  rule-based modeling}}}.
\newblock {\sl \bibinfo{journal}{Bioinformatics}}
  \bibinfo{volume}{34}(\bibinfo{number}{13}), pp. \bibinfo{pages}{i583--i592},
  \doi{10.1093/bioinformatics/bty272}.

\bibitemdeclare{article}{Braatz:2010aa}
\bibitem{Braatz:2010aa}
\bibinfo{author}{B. \surnamestart Braatz\surnameend},
  \bibinfo{author}{H. \surnamestart Ehrig\surnameend},
  \bibinfo{author}{K. \surnamestart Gabriel\surnameend} \&
  \bibinfo{author}{U. \surnamestart Golas\surnameend}
  (\bibinfo{year}{2014}): \emph{\bibinfo{title}{{Finitary $\mathcal{M}$
  -adhesive categories}}}.
\newblock {\sl \bibinfo{journal}{MSCS}}
  \bibinfo{volume}{24}(\bibinfo{number}{4}), pp.
  \bibinfo{pages}{240403--240443}, \doi{10.1017/S0960129512000321}.

\bibitemdeclare{inproceedings}{Corradini2006}
\bibitem{Corradini2006}
\bibinfo{author}{A. \surnamestart Corradini\surnameend},
  \bibinfo{author}{T. \surnamestart Heindel\surnameend},
  \bibinfo{author}{F. \surnamestart Hermann\surnameend} \&
  \bibinfo{author}{B. \surnamestart K\"{o}nig\surnameend}
  (\bibinfo{year}{2006}): \emph{\bibinfo{title}{{Sesqui-Pushout Rewriting}}}.
\newblock In: {\sl
  \bibinfo{booktitle}{Graph Transformations (ICGT 2006)}}, {\sl
  \bibinfo{series}{LNCS}} \bibinfo{volume}{4178}, pp. \bibinfo{pages}{30--45},
  \doi{10.1007/11841883_4}.

\bibitemdeclare{incollection}{Danos2014}
\bibitem{Danos2014}
\bibinfo{author}{V. \surnamestart Danos\surnameend},
  \bibinfo{author}{R. \surnamestart Heckel\surnameend} \&
  \bibinfo{author}{P. \surnamestart Sobocinski\surnameend}
  (\bibinfo{year}{2014}): \emph{\bibinfo{title}{{Transformation and Refinement
  of Rigid Structures}}}.
\newblock In: {\sl \bibinfo{booktitle}{Graph Transformation (ICGT 2014)}}, {\sl
  \bibinfo{series}{LNCS}} \bibinfo{volume}{8571},  pp. \bibinfo{pages}{146--160},
  \doi{10.1007/978-3-319-09108-2_10}.

\bibitemdeclare{article}{ehrig:2006aa}
\bibitem{ehrig:2006aa}
\bibinfo{author}{H.~\surnamestart Ehrig\surnameend},
  \bibinfo{author}{K.~\surnamestart Ehrig\surnameend},
  \bibinfo{author}{U.~\surnamestart Prange\surnameend} \&
  \bibinfo{author}{G.~\surnamestart Taentzer\surnameend}
  (\bibinfo{year}{2006}): \emph{\bibinfo{title}{Fundamentals of Algebraic Graph
  Transformation}}.
\newblock {\sl \bibinfo{journal}{Monographs in Theoretical Computer Science. An
  EATCS Series}}, \doi{10.1007/3-540-31188-2}.

\bibitemdeclare{article}{ehrig2014mathcal}
\bibitem{ehrig2014mathcal}
\bibinfo{author}{H. \surnamestart Ehrig\surnameend},
  \bibinfo{author}{U. \surnamestart Golas\surnameend},
  \bibinfo{author}{A. \surnamestart Habel\surnameend},
  \bibinfo{author}{L. \surnamestart Lambers\surnameend} \&
  \bibinfo{author}{F. \surnamestart Orejas\surnameend}
  (\bibinfo{year}{2014}): \emph{\bibinfo{title}{{$\mathcal{M}$-adhesive
  transformation systems with nested application conditions. Part 1:
  parallelism, concurrency and amalgamation}}}.
\newblock {\sl \bibinfo{journal}{MSCS}}
  \bibinfo{volume}{24}(\bibinfo{number}{04}), \doi{10.1017/s0960129512000357}.

\bibitemdeclare{book}{FlajoletSedgewick}
\bibitem{FlajoletSedgewick}
\bibinfo{author}{P. \surnamestart Flajolet\surnameend} \&
  \bibinfo{author}{R. \surnamestart Sedgewick\surnameend}
  (\bibinfo{year}{2009}): \emph{\bibinfo{title}{Analytic Combinatorics}}.
\newblock \bibinfo{publisher}{Cambridge University Press},
  \doi{10.1017/cbo9780511801655}.

\bibitemdeclare{article}{Gillespie1977}
\bibitem{Gillespie1977}
\bibinfo{author}{D.T. \surnamestart Gillespie\surnameend}
  (\bibinfo{year}{1977}): \emph{\bibinfo{title}{Exact stochastic simulation of
  coupled chemical reactions}}.
\newblock {\sl \bibinfo{journal}{The Journal of Physical Chemistry}}
  \bibinfo{volume}{81}(\bibinfo{number}{25}), pp. \bibinfo{pages}{2340--2361},
  \doi{10.1021/j100540a008}.

\bibitemdeclare{article}{habel2009correctness}
\bibitem{habel2009correctness}
\bibinfo{author}{A. \surnamestart Habel\surnameend} \&
  \bibinfo{author}{K.-H. \surnamestart Pennemann\surnameend}
  (\bibinfo{year}{2009}): \emph{\bibinfo{title}{Correctness of high-level
  transformation systems relative to nested conditions}}.
\newblock {\sl \bibinfo{journal}{MSCS}}
  \bibinfo{volume}{19}(\bibinfo{number}{02}), p. \bibinfo{pages}{245},
  \doi{10.1017/s0960129508007202}.

\bibitemdeclare{article}{Joyal1981}
\bibitem{Joyal1981}
\bibinfo{author}{A. \surnamestart Joyal\surnameend}
  (\bibinfo{year}{1981}): \emph{\bibinfo{title}{{Une th\'{e}orie combinatoire
  des s\'{e}ries formelles}}}.
\newblock {\sl \bibinfo{journal}{Advances in Mathematics}}
  \bibinfo{volume}{42}(\bibinfo{number}{1}), pp. \bibinfo{pages}{1--82},
  \doi{10.1016/0001-8708(81)90052-9}.

\bibitemdeclare{book}{norris}
\bibitem{norris}
\bibinfo{author}{J.R. \surnamestart Norris\surnameend}
  (\bibinfo{year}{1998}): \emph{\bibinfo{title}{{Markov Chains}}}.
\newblock \bibinfo{series}{Cambridge Series in Statistical and Probabilistic
  Mathematics}, \bibinfo{publisher}{Cambridge University Press}.

\bibitemdeclare{article}{remy1985procede}
\bibitem{remy1985procede}
\bibinfo{author}{J.-L. \surnamestart R{\'e}my\surnameend}
  (\bibinfo{year}{1985}): \emph{\bibinfo{title}{Un proc{\'e}d{\'e} it{\'e}ratif
  de d{\'e}nombrement d'arbres binaires et son application {\`a} leur
  g{\'e}n{\'e}ration al{\'e}atoire}}.
\newblock {\sl \bibinfo{journal}{RAIRO. Informatique th{\'e}orique}}
  \bibinfo{volume}{19}(\bibinfo{number}{2}), pp. \bibinfo{pages}{179--195},
  \doi{10.1051/ita/1985190201791}.

\bibitemdeclare{article}{Rowland2010}
\bibitem{Rowland2010}
\bibinfo{author}{E.S. \surnamestart Rowland\surnameend}
  (\bibinfo{year}{2010}): \emph{\bibinfo{title}{Pattern avoidance in binary
  trees}}.
\newblock {\sl \bibinfo{journal}{Journal of Combinatorial Theory, Series A}}
  \bibinfo{volume}{117}(\bibinfo{number}{6}), pp. \bibinfo{pages}{741--758},
  \doi{10.1016/j.jcta.2010.03.004}.

\end{thebibliography}
\end{document}